%% file: arxiv.tex
\documentclass[blackref,approvalform, 11pt]{article}
\usepackage{fullpage}
\usepackage[utf8]{inputenc}
\usepackage{mathpazo}
\usepackage{natbib}
\usepackage{booktabs} 

\usepackage{booktabs} 
\usepackage[linesnumbered,noend,ruled,noline]{algorithm2e} 

\SetAlFnt{\small}
\SetAlCapFnt{\small}
\SetAlCapNameFnt{\small}
\SetAlCapHSkip{0pt}
\IncMargin{-\parindent}
\SetArgSty{textnormal}
 
\SetCommentSty{mycommfont}

\input{packages}

\title{Forging Self-Funded Marketplaces among Strategic Agents}
\date{}
\author{
Yuan Deng\\
{Google Research}\\
\small\href{mailto:dengyuan@google.com}{dengyuan@google.com}
\and
Vasilis Gkatzelis\\
{Drexel University}\\
\small\href{mailto:gkatz@drexel.edu}{gkatz@drexel.edu}
\and
Xizhi Tan\\
{Stanford University}\\
\small\href{mailto:xizhi@stanford.edu}{xizhi@stanford.edu}
\and
Grigoris Velegkas\\
{Google Research}\\
\small\href{mailto:gvelegkas@google.com}{gvelegkas@google.com}
\and 
Song Zuo\\
{Google Research}\\
\small\href{mailto:szuo@google.com}{szuo@google.com}
}
\begin{document}
\maketitle

\begin{abstract}
We introduce the problem of designing mechanisms that incentivize strategic agents to form self-funded  marketplaces. In our model, if agent $i$ exerts effort $x_i\in [0,1]$, they incur a cost of $x_i\cdot c_i$ (where $c_i$ is unknown to the mechanism designer) and they generate revenue $x_i\cdot r_i$; crucially, $c_i$ can be greater or smaller than $r_i$. Each effort profile $\mathbf{x}$ yields value $v(\mathbf{x})$ and the objective is to choose an effort vector that maximizes the value while ensuring that every agent $i$ receives a payment $p_i\geq x_i\cdot c_i$ and that $\mathbf{x}$ is budget-balanced, i.e., $\sum_{i} p_i \leq \sum_{i} x_i\cdot r_i$. This problem generalizes the well-studied budget-feasible mechanism design problem, where the requirement is that  $\sum_{i} p_i \leq B$ for some predetermined budget $B$.

To evaluate the performance of such mechanisms, we first consider the first-best benchmark (the optimal value achievable in the absence of any private information) and show that no truthful auction can achieve a bounded approximation of this benchmark. Also, even in restricted settings, no auction can achieve better than a logarithmic approximation. We complement these results by proposing a class of sequential auctions whose subgame perfect equilibria guarantee a logarithmic approximation of this benchmark. We then introduce an alternative benchmark, the maximin share (MMS), that better captures the thickness of the market and we provide an auction whose subgame perfect equilibria achieve a constant approximation of this benchmark. 
\end{abstract}

\input{notation}

\input{intro}

\input{ourresult}

\input{relatedwork}

\input{prelim}

\input{approxopt}

\input{bafo_results}

\input{mms}

\input{conclusion}

\bibliographystyle{alpha}
\bibliography{refs}
\newpage
\appendix
\crefalias{section}{appendix}
\crefalias{subsection}{appendix}

\input{apx_bafo}

\input{apx_mms}

\end{document}

%% file: packages.tex
\usepackage{graphicx}
\usepackage{xcolor}
\definecolor{DrexelBlue}{HTML}{003058}
\definecolor{commentblue}{HTML}{2A6EC7}

\usepackage{hyperref}
\hypersetup{
  colorlinks = true,  
  urlcolor = {commentblue},
  linkcolor = {DrexelBlue}, 
  citecolor = {YellowOrange!85!black}
}
\usepackage{booktabs} 
\usepackage{multirow} 
\usepackage{multicol} 
\usepackage{natbib}
\usepackage{color-edits} 
\addauthor[Grigoris]{gv}{teal}
\addauthor[Yuan]{yd}{pink}
\addauthor[Xizhi]{xt}{gold}
\addauthor[Vasilis]{vg}{red}

\usepackage[utf8]{inputenc}

\usepackage{amsmath}
\usepackage{amsthm}
\usepackage{bbm}
\usepackage{blkarray}
\usepackage{color}
\usepackage{enumerate}
\usepackage{float} 
\usepackage{subcaption}
\usepackage{tikz}
\usetikzlibrary{arrows.meta, positioning, calc, shapes, fit}
\usepackage{mathtools}
\usepackage[inline]{enumitem} 
\usepackage{nicefrac}
\usepackage{thm-restate} 
\usepackage{centernot}

\usepackage[capitalise,noabbrev,nameinlink]{cleveref} 

\usepackage[dvipsnames]{xcolor}
\definecolor{niceRed}{RGB}{190,38,38}
\definecolor{blueGrotto}{HTML}{059DC0}
\definecolor{royalBlue}{HTML}{057DCD}
\definecolor{navyBlue}{HTML}{0B579C}
\definecolor{limeGreen}{HTML}{81B622}
\definecolor{nicePurple}{HTML}{9c27b0}
\definecolor{lightRoyalBlue}{HTML}{def2ff} 
 \definecolor{gold}{HTML}{ffa300}

\usepackage{soul} 

\newtheorem{theorem}{Theorem}[section]
\newtheorem{corollary}[theorem]{Corollary}

\newtheorem{lemma}[theorem]{Lemma}

 \newtheorem*{inftheorem*}{Informal Theorem}

\newtheorem{definition}[theorem]{Definition}

\newtheorem*{definition*}{Definition}

\theoremstyle{remark} 

\newtheorem{remark}{Remark}

\AfterEndEnvironment{definition}{\noindent\ignorespaces}
\AfterEndEnvironment{infdefinition}{\noindent\ignorespaces}
\AfterEndEnvironment{example}{\noindent\ignorespaces}
\AfterEndEnvironment{assumption}{\noindent\ignorespaces}
\AfterEndEnvironment{lemma}{\noindent\ignorespaces}
\AfterEndEnvironment{theorem}{\noindent\ignorespaces}
\AfterEndEnvironment{proposition}{\noindent\ignorespaces}
\AfterEndEnvironment{fact}{\noindent\ignorespaces}
\AfterEndEnvironment{question}{\noindent\ignorespaces}
\AfterEndEnvironment{corollary}{\noindent\ignorespaces}
\AfterEndEnvironment{model}{\noindent\ignorespaces}
\AfterEndEnvironment{remark}{\noindent\ignorespaces}
\AfterEndEnvironment{proof}{\noindent\ignorespaces}
\AfterEndEnvironment{fact}{\noindent\ignorespaces}
\AfterEndEnvironment{minftheorem}{\noindent\ignorespaces}
\AfterEndEnvironment{inftheorem}{\noindent\ignorespaces}
\AfterEndEnvironment{maintheorem}{\noindent\ignorespaces}
\AfterEndEnvironment{restatable}{\noindent\ignorespaces}
\AfterEndEnvironment{observation}{\noindent\ignorespaces}

\crefname{section}{Section}{Sections}
\crefname{theorem}{Theorem}{Theorems}
\crefname{theorem*}{Theorem}{Theorems}
\crefname{inftheorem}{Informal Theorem}{Informal Theorems}
\crefname{assumption}{Assumption}{Assumptions}
\crefname{lemma}{Lemma}{Lemmas}
\crefname{definition}{Definition}{Definitions}
\crefname{infdefinition}{Informal Definition}{Informal Definitions}
\crefname{conjecture}{Conjecture}{Conjectures}
\crefname{corollary}{Corollary}{Corollaries}
\crefname{construction}{Construction}{Constructions}
\crefname{conjecture}{Conjecture}{Conjectures}
\crefname{claim}{Claim}{Claims}
\crefname{observation}{Observation}{Observations}
\crefname{proposition}{Proposition}{Propositions}
\crefname{fact}{Fact}{Facts}
\crefname{question}{Question}{Questions}
\crefname{problem}{Problem}{Problems}
\crefname{remark}{Remark}{Remarks}
\crefname{example}{Example}{Examples}
\crefname{equation}{Equation}{Equations}
\crefname{appendix}{Appendix}{Appendices}
\crefname{model}{Model}{Models}
\crefname{figure}{Figure}{Figures}
\crefname{condition}{Condition}{Conditions}
\crefname{algorithm}{Mechanism}{Mechanisms}
\Crefname{algorithm}{Mechanism}{Mechanisms}

\usepackage{etoolbox}

\newcommand{\eat}[1]{}

\makeatletter
\newcommand{\tagnum}[2]{%
    \refstepcounter{equation}%
    \tag{#1) \ (\theequation}%
    \protected@write \@auxout {}{%
        \string \newlabel {#2}{{\theequation}{\thepage}{}{equation.\theequation}{}}%
    }%
}
\makeatother

\newcommand{\R}{\mathbb{R}}

\newcommand{\E}{\operatornamewithlimits{\mathbb{E}}}

\newcommand{\OPT}{\mathrm{OPT}}

\renewcommand{\epsilon}{\varepsilon}

\makeatletter
\newcommand*{\tran}{{\mathpalette\@tran{}}}
\newcommand*{\@tran}[2]{\raisebox{\depth}{$\m@th#1\intercal$}}
\makeatother

\renewcommand{\bar}{\overline}

\def\<{\langle}
\def\>{\rangle}

\DeclareMathAlphabet{\mathpzc}{OT1}{pzc}{m}{it}

\DeclareMathAlphabet{\mathdutchcal}{U}{dutchcal}{m}{n}
\SetMathAlphabet{\mathdutchcal}{bold}{U}{dutchcal}{b}{n}
\DeclareMathAlphabet{\mathdutchbcal}{U}{dutchcal}{b}{n}
 
\DeclareMathAlphabet\urwscr{U}{urwchancal}{b}{n}%
\DeclareMathAlphabet\rsfscr{U}{rsfso}{m}{n}
\DeclareMathAlphabet\euscr{U}{eus}{m}{n}
\DeclareFontEncoding{LS2}{}{}
\DeclareFontSubstitution{LS2}{stix}{m}{n}
\DeclareMathAlphabet\stixcal{LS2}{stixcal}{m} {n}

\newcommand{\opt}{\textrm{OPT}}

        \usepackage{tikz}
\usetikzlibrary{patterns}

\renewcommand{\hat}{\widehat}

\usepackage{array}
\newcolumntype{L}[1]{>{\raggedright\let\newline\\\arraybackslash\hspace{0pt}}m{#1}}
\newcolumntype{C}[1]{>{\centering\let\newline\\\arraybackslash\hspace{0pt}}m{#1}}
\newcolumntype{R}[1]{>{\raggedleft\let\newline\\\arraybackslash\hspace{0pt}}m{#1}}

\usepackage{newpxtext}
\usepackage{newpxmath}

%% file: notation.tex
\newcommand{\alloc}{\mathbf{x}}
\newcommand{\pay}{\mathbf{p}}
\newcommand{\bb}{budget-balance}
\newcommand{\wel}{\mathcal{W}}
\newcommand{\cost}{\mathbf{c}}
\newcommand{\revenue}{\mathbf{r}}
\newcommand{\mms}{\text{MMS}}
\newcommand{\sumval}{\mathcal{V}}
\newcommand{\ladderbafo}{\textsc{BAFOLadder}}
\newcommand{\bids}{\mathbf{b}}
\newcommand{\guessbafo}{\textsc{Guess-and-BAFO}}
\newcommand{\fb}{\text{FB}}
\newcommand{\pred}{\hat{\opt}}

%% file: intro.tex
\section{Introduction}

Consider a set $N$ of content creators and a function $v(\cdot)$ that quantifies the collective value that these agents can generate if each agent $i$ exerts some effort $x_i\in [0,1]$. Apart from value, each content creator $i$ also generates some revenue $x_i\cdot r_i$ through their content, but they suffer a cost $x_i \cdot c_i$ for their effort, where $c_i$ is private. We approach this setting from the perspective of a principal, such as a platform owner, wants to choose an effort profile $\mathbf{x}$ that maximizes $v(\mathbf{x})$. However, the principal faces two important obstacles. The first is that the content creators are \emph{strategic} agents and their costs are \emph{private}. Therefore, to appropriately incentivize them, the principal must design an auction-based mechanism that determines the effort vector $\mathbf{x}$ and the payment $p_i \geq x_i\cdot c_i$ for each $i\in N$. The second obstacle is that $\mathbf{x}$ needs to be \emph{self-funded}: the total payment required to compensate the content creators must be covered by the aggregate revenue they generate, i.e., $\sum_{i\in N} p_i \leq \sum_{i\in N} x_i \cdot r_i$.

This \emph{autarkic mechanism design} problem poses a formidable challenge. For example, it generalizes the well-studied problem of budget-feasible mechanism design introduced by Singer~\cite{singer2010budget}, where a principal seeks to maximize $v(\mathbf{x})$ subject to an exogenous budget constraint $\sum_{i\in N}p_i \leq B$. The inherent difficulty in the budget-feasible setting stems from the fact that  the feasibility constraint is \emph{private}, so the principal has no estimate regarding what value they should be aiming for: e.g., if $\sum_{i\in N} c_i \leq B$, it would even be feasible to compensate every $i\in N$ for maximum effort, but the principal does not know the agents' costs. Since the agents are strategic and their goal is to maximize their payment, eliciting their costs and satisfying the budget constraint becomes more demanding. This difficulty is compounded in our setting: instead of an exogenously defined budget, the principal needs to satisfy an endogenous budget constraint determined by the choice of $\mathbf{x}$ itself. 

The fact that the budget is endogenous in the autarkic mechanism design problem introduces \emph{complementarities} across the agents, and it gives rise to a feasibility constraint that \emph{may not even be downward-closed}. For example, consider an instance of two content creators with costs $c_1 = 1,~ c_2=1$, and revenues $r_1=2,~ r_2=0$. Note that hiring both of them is feasible, as it would generate a total revenue of $2$, which is sufficient to cover their costs. However, hiring the second agent alone is infeasible. This is in contrast to budget-feasible mechanism design where the budget is exogenous, leading to a downward-closed feasibility constraint. To demonstrate how this complementarity complicates the problem, assume that the second agent can generate much more value than the first one. To maximize the value, the principal would like to hire both of them, but why would agent $1$ be willing to subsidize the cost of agent $2$? Since the costs are private, agent $1$ could pretend that their cost is $2$ and the principal would still prefer to hire them rather than hiring no one.

In light of these obstacles, what would be an appropriate benchmark for evaluating the performance of autarkic mechanisms? The most ambitious one is the first-best benchmark, i.e., the optimal value achievable if the agents' costs were public: $\max_{\mathbf{x}} v(\mathbf{x})$ subject to $\sum_{i\in N} x_i \cdot c_i\leq \sum_{i\in N} x_i\cdot r_i$. However, this may be too demanding, especially in instances like the one above, where maximizing the value heavily depends on just one agent. Our goal in this paper is to study this natural problem and to evaluate the extent to which incentive compatible mechanisms can approximate appropriate value benchmarks.

\begin{center}
\emph{What are appropriate value benchmarks for evaluating autarkic mechanisms, and to what\\ extent can we approximate these benchmarks despite the strategic behavior of the agents?}
\end{center}

%% file: ourresult.tex
\subsection{Our Results} 

We introduce the autarkic mechanism design problem, which is interesting from both a practical and a theoretical perspective. On one hand, it captures a variety of modern marketplaces that focus on self-sufficiency (see \cref{sec:relatedWork} for a discussion regarding applications). On the other hand, it poses new technical challenges due to its complicated feasibility constraint that is simultaneously \emph{private} and \emph{not downward-closed}. We focus on instances where $v(\mathbf{x})$, where $v: [0, 1]^n \to \mathbb{R}_{\ge 0}$ is an arbitrary non-decreasing subadditive function. Note that although we assume that each agent can exert any amount of effort $x_i\in [0,1]$ almost all of our results also hold for the ``integral'' case where $x_i\in \{0, 1\}$ (see \cref{sec:conclusion} for a more detailed discussion).

\paragraph{Inapproximability of the First-Best Benchmark  (\cref{sec:opt-lower-bounds}).}
To better understand the hardness of this problem, we first focus on the extent to which incentive-compatible mechanisms can approximate the first-best benchmark. For dominant-strategy incentive-compatible (i.e., truthful) mechanisms, we show that no bounded approximation is achievable, even if the mechanism designer knows the optimal value, \opt, of the first-best benchmark in the instance at hand.

\begin{inftheorem*}
    No deterministic or randomized truthful-in-expectation mechanism can achieve a bounded approximation to the first-best benchmark, even if the mechanism knows the optimal value $\opt$.
\end{inftheorem*}

At the core of this impossibility result is a class of two-agent instances similar to the one discussed in the introduction. Notably, it shows that even if we know exactly what the value of $\opt$ is, the fact that the costs of the two agents are private makes it impossible for a mechanism to elicit the true costs using payments that do not exceed the ``budget'' that the revenue generates. 


Motivated by this observation, we further restrict these instances and consider the case where the cost of one of the two agents is public. If $\opt$ is known in advance, this reduction to ``single-agent'' instances makes the problem easy. However, we show that if the mechanism is instead provided with a (normalized) range $[1,\mathcal{V}]$ within which $\opt$ lies, and we evaluate the achievable approximation as a function of $\mathcal{V}$, the performance of any mechanism grows logarithmically in $\mathcal{V}$.

\begin{inftheorem*}
For every deterministic or randomized mechanism and any $n\geq 2$,
there exists an instance with a single strategic agent (the costs of all other agents are public) such that, even if the mechanism knows that $\opt$ lies in some range $[1,\mathcal{V}]$, its approximation of $\opt$ is $\Omega(\log \mathcal{V})$ in any equilibrium.
\end{inftheorem*}

\paragraph{Subgame Perfect Equilibria and Best and Final Offers  (\cref{sec:bafo}).}
Complementing these negative results, we then propose a sequential auction which, given a range $[1,\sumval]$ for \opt, achieves an optimal approximation of $O(\log\sumval)$ in any subgame perfect equilibrium. This auction randomly chooses a target value $\overline{v}$ within $[1,\sumval]$ and it commits to only returing a feasible effort profile $\mathbf{x}$ if the value it generates satisfies that threshold, i.e., $v(\mathbf{x})\geq \overline{v}$. Then, given that commitment, it uses the \emph{best and final offer} (BAFO) process, analyzed by Gkatzelis et al.~\cite{GML25}: it sequentially approaches the agents, asking them to report the payment they request, given the payments requested by all previous agents.\footnote{Note that this can also be implemented as a descending auction with personalized prices $p_i$ for each agent $i$ that drop over time. Before each price drop of $p_i$, agent $i$ can either accept the drop or permanently ``freeze'' at the current price (see~\cite{GML25}).}
After collecting all the bids $(b_i)_{i\in N}$, the principal selects some $\mathbf{x}$ with $v(\mathbf{x})\geq \overline{v}$ that is budget-feasible with respect to their bids, i.e., $\sum_{i\in N}x_i\cdot b_i\leq \sum_{i\in N} x_i\cdot r_i$, and compensates each agent $i\in N$ with a payment $x_i\cdot b_i$, i.e., proportional to their bid and effort.

\begin{inftheorem*}
    The {\guessbafo} auction, provided with a range $[1,\sumval]$ for $\OPT$, guarantees an $O(\log \sumval)$ approximation to the first-best benchmark in every subgame perfect equilibrium.
\end{inftheorem*}

The success of this auction is based on the observation that if the value target that the principal ``guessed'' is feasible, then the commitment to only hire if this target is met introduces competition among the agents and forces them to report more ``reasonable'' prices to avoid being excluded. In fact, this is just one instance of a broader ``wish-come-true'' approach, which may be of independent interest: if the principal defines a selection rule based on \emph{any} price-monotonic condition (value, cardinality, diversity constraints, etc.) that is feasible under the agents' true costs, then the sequential BAFO process ensures that the agents will compete and coordinate to satisfy this condition in every subgame perfect equilibrium. 


\paragraph{The Maximin Share (MMS) benchmark  (\cref{sec:mms}).}
Our results on the first-best benchmark demonstrate that it does not capture an important parameter of the instance at hand: the amount of competition. Two instances with the same \opt\ value can be very different in terms of their ``strategic'' complexity, depending on the amount of competition. For example, an instance where a single agent brings all the revenue may have the same optimal value as one with multiple such agents, even though the former can be much more demanding due to the lack of competition. 
To address this issue, we introduce the Maximin Share (MMS) benchmark,\footnote{The name of this benchmark is a tribute to the closely related MMS fairness notion introduced by Budish~\cite{Budish11} to address the fact that classic fairness notions like proportionality and envy-freeness may not be achievable when allocating indivisible goods.} which better captures the extent to which a few agents may have a lot of control over the total value.
To define this benchmark, consider any partition {of} the set $N$ of agents into two disjoint groups $P_1$ and $P_2$ and let $\mathcal{W}(P_1)$ and $\mathcal{W}(P_2)$ denote the optimal budget-balanced values achievable using only agents in $P_1$ and $P_2$, respectively. Then, the MMS benchmark is equal to $\max_{P_1,P_2}\min\{\mathcal{W}(P_1), \mathcal{W}(P_2)\}$. In other words, for each partition of $N$ into $P_1$ and $P_2$, it focuses on the optimal value of the second-best group, and then chooses the partition that maximizes this quantity. Note that for instances where the contribution to the optimal value and revenue is not dominated by a few agents, the MMS benchmark closely approximates $\opt$. For example, if every agent's marginal contribution to the value and revenue is negligible, the MMS benchmark converges to half of \opt\ because there exists a partition such that $\mathcal{W}(P_1)\approx \mathcal{W}(P_2) \approx \opt/2$. Therefore, this provides a natural second-best benchmark for autarkic mechanism design.  

We first find that truthfulness remains a significant hurdle even under this benchmark. 
\begin{inftheorem*}
No deterministic or randomized truthful-in-expectation mechanism can achieve an approximation ratio better than $\Omega(\log n)$ against the $\mms$ benchmark.
\end{inftheorem*}

However, our main result in this case is a positive one, even without any estimates regarding the optimal value. We propose the \ladderbafo, a sequential auction that achieves a constant approximation to the MMS benchmark in every subgame perfect equilibrium. 

\begin{inftheorem*} The \ladderbafo\ mechanism achieves a constant-factor approximation to the MMS benchmark in every subgame perfect equilibrium. \end{inftheorem*}

Rather than using an estimate of \opt\ to introduce competition among the agents, the \ladderbafo\ mechanism randomly splits them into two groups and lets the two groups compete with each other. Specifically, it uses each group to determine a value target for the other one, and alternates between them until one of the two groups is unable to meet the latest target.

%% file: relatedwork.tex
\subsection{Related Work}\label{sec:relatedWork}

\paragraph{Self-Funded Markets.}
Our model is motivated by platforms that rely on third-party content creators or data owners---the \emph{agents} in our model---where the payments to these agents are funded directly by the revenue their content generates within the platform. We capture these interactions by letting \(x_i \in [0,1]\) denote agent \(i\)'s effort, content, or data, which is procured by the platform.  
Procuring this content imposes a private cost \(x_i \cdot c_i\) on the agent.
Crucially, these settings highlight that an agent's contribution has two distinct dimensions that need not be aligned: \emph{value} and \emph{revenue}. The platform's overall value, \(v(\mathbf{x})\), captures broad objectives such as output quality, coverage, diversity, or long-term user trust. Conversely, revenue \(x_i \cdot r_i\) is the tangible income associated with the agent's content, such as advertising, subscriptions, or user engagement.

In AI search systems that utilize retrieval-augmented generation (RAG), when a user poses a query, the platform synthesizes an answer by drawing on multiple external web sources (the agents). The platform derives \emph{value} from these sources because they improve the accuracy and quality of the generated answer. However, the \emph{revenue} (e.g., from ads displayed alongside the search result \cite{google_ai_features,openai_ads}) might not perfectly align with the value provided. For instance, a highly specialized medical blog might contribute crucial \emph{value} to a complex query but generate little direct ad revenue. Emerging initiatives, such as Index by Parallel \cite{parallel_index}, seek to explicitly compensate these content owners based on their contribution to the AI-agent outputs.

Similar dynamics arise on creator platforms \cite{tiktok_rewards,youtube_partner,spotify_royalties,medium_partner}. Agents (e.g., video creators, musicians, or writers) incur production costs to make content. Broadly popular, commercially driven content may generate high immediate \emph{revenue} via streams and ads, while specialized educational content might yield low revenue but significantly boost the platform's catalog diversity and reputation (\emph{value}).

In all these examples, selecting different agents changes both the outcome produced and the funds available for payments. This introduces our core technical challenge, an \emph{endogenous budget}: unlike standard procurement, which starts with a fixed budget \(B\) independent of the selected allocation \cite{singer2010budget}, our model requires that total payments must be covered strictly by the revenue generated by that specific allocation, \(\sum_i r_i x_i\).

\paragraph{The Best and Final Offer (BAFO) Protocol.}
In the sequential mechanisms that we propose, each bidder reports their best-and-final-offer, which is a standard practice in government and high-value commercial contracting~\cite{uslegal_bafo}. These auctions can also be implemented as descending auctions where each bidder can choose to ``freeze'' their price, signaling that this is their best-and-final-offer. While auctions using BAFO strategies are widely adopted in practice, not much was known about the theoretical guarantees they can provide. Gkatzelis et al.~\cite{GML25} initiated the theoretical analysis of procurement auctions with BAFO strategies, focusing on social welfare maximization,  and they proposed mechanisms using BAFO that achieve the optimal welfare in every subgame perfect equilibrium. 
We use the same sequential commitment
device for a different purpose: to coordinate the agents on an allocation that
satisfies an endogenous, coalition-dependent budget constraint while attaining
a target value.

\paragraph{Budget-Feasible Mechanism Design.}
A closely related problem is budget-feasible mechanism design, introduced by Singer~\cite{singer2010budget}, in which a principal maximizes value subject to
an exogenous hard budget and sellers have private costs. For additive
valuations, Chen et al.~\cite{chen2011approximability} obtained deterministic
and randomized approximation factors of $2+\sqrt{2}$ and $3$, respectively. Gravin et al.~\cite{gravin2019optimal} subsequently obtained tight factors of
$3$ for deterministic mechanisms and $2$ for randomized mechanisms.

For monotone submodular valuations, Chen et al.~\cite{chen2011approximability} gave approximation factors of $8.34$
deterministically and $7.91$ in expectation. Jalaly and Tardos~\cite{JT21} obtained a polynomial-time randomized factor of $5$;
with unrestricted computation, they also obtained factors of $4$ for
randomized mechanisms and approximately $4.56$ for deterministic mechanisms.
Among computationally efficient mechanisms, Balkanski et al.~\cite{balkanski2022deterministic} gave the first polynomial-time
deterministic constant-factor mechanism, a descending-clock auction with
factor $4.75$. The TripleEagle clock framework of Han et al.~\cite{han2023triple} achieved factors of approximately $4.3$ for
randomized mechanisms and $4.45$ for deterministic mechanisms using a linear
number of value queries. A subsequent randomized mechanism improved the former
factor to $4.08$ \cite{han2025efficient}. Most recently,  Han and Lu~\cite{hanlu2026improved} give linear-time
deterministic and randomized mechanisms with approximation factor $3.798$.

Constant-factor mechanisms are also known for non-monotone submodular
valuations \cite{amanatidis2019budget}. For XOS and subadditive valuations,
stronger demand-oracle models play a central role
\cite{dobzinski2011mechanisms,bei2017worst}. Neogi et al.~\cite{neogi2024budget} obtained a polynomial-time constant-factor mechanism
for XOS valuations with demand-oracle access, as well as an explicit
constant-factor mechanism for general subadditive valuations without a
polynomial-time guarantee. The polynomial-time demand-oracle approximation for
general subadditive valuations was subsequently improved to
$O(\log\log n)$ \cite{neogi2025subadditive}.

Departing from truthful mechanism design, Bhawalkar et al.~\cite{BDLMP25} analyze a simultaneous first-price, or pay-as-bid,
budget-feasible procurement marketplace. Under a small-seller condition, their
equilibrium-efficiency guarantee converges to a factor of $2$ in the
large-market limit.

More broadly, many foundational budget-feasible mechanisms use direct
sealed-bid formats
\cite{singer2010budget,chen2011approximability,JT21}, whereas a parallel line
studies dynamic, descending-clock, and deferred-acceptance implementations.
Early examples considered unit valuations or knapsack-type procurement
\cite{ensthaler2014dynamic,JM17}; more recent mechanisms cover monotone
submodular objectives
\cite{balkanski2022deterministic,huang2023randomized,han2023triple,
han2025efficient,hanlu2026improved}. Another line obtains stronger guarantees
under large-market assumptions, in which no individual seller is pivotal
relative to the available budget or total value
\cite{anari2014mechanism,JT21,BDLMP25}. The central distinction from all of
this work is that their budget is exogenous. In our model, the available
budget depends on the revenue generated by the selected agents, so the
resulting feasibility constraint need not be downward-closed.

\paragraph{Budget-Balanced Mechanisms.}
In a cost-sharing problem, a mechanism selects a set of agents to receive a
jointly supplied service and charges those agents so that the collected
payments cover, exactly or approximately, the cost of serving them.
The autarkic market problem could be thought of as a ``dual problem'' to the budget-balanced cost-sharing problem. Despite the difference in the direction of transfers, both settings are governed by the Green–Kohlberg–Laffont trilemma \cite{GKL76}, which establishes the fundamental incompatibility of economic efficiency, budget balance (BB), and strategyproofness. For any mechanism satisfying both BB and incentive compatibility, the multiplicative approximation of welfare can be unbounded \cite{feigenbaum2003hardness}. Consequently, much of the literature focuses on Moulin mechanisms \cite{moulin1999incremental, moulin2001strategyproof}, which leverage cross-monotonic cost-sharing rules to guarantee group-strategyproofness and coalitional stability. 
The worst-case efficiency loss of such mechanisms is strictly determined by the combinatorial ``summability'' of the cost function \cite{roughgarden2009quantifying}, though welfare can be improved by relaxing exact BB to allow for small deficits or almost-balanced outcomes \cite{moulin2001strategyproof}.

\paragraph{General Procurement.} The study of frugal mechanism design is motivated by the challenge of minimizing unnecessary overpayment while preserving truthfulness in procurement auctions and combinatorial team selection problems. This perspective has been applied across a variety of classic optimization settings. In network design, Archer and Tardos~\cite{archer2002frugal} and Talwar~\cite{talwar2003price} initiated the study of frugality in path auctions, demonstrating that VCG mechanisms can incur large overpayments and establishing tight frugality bounds under the frugal solution benchmark. For coverage problems such as vertex cover and set cover, Elkind et al.~\cite{elkind2007frugality} and subsequent works developed truthful mechanisms with frugality ratios that depend on structural parameters like graph degree or spectral properties, and also provided matching lower bounds in several cases. Frugality in $k$-flow and cut problems has also been explored, with mechanisms achieving constant-competitive frugality ratios relative to various benchmarks, including both disjoint-alternative and equilibrium-based definitions~\cite{karlin2005beyond, kempe2010spectral, chen2010frugal}. In some settings, such as matroid and spanning tree auctions, VCG mechanisms are provably frugal, achieving optimal or near-optimal frugality ratios~\cite{karlin2005beyond,talwar2003price}. Recently, Balkanski et al.~\cite{balkanski2025procurement} considered the frugality ratio of VCG in the strategic facility location problem and proved a tight bound of 3.
Apart from cost minimization, recent work has considered maximizing different objective functions of the buyer, such as the ``gains from trade'' (GFT)—the difference between the buyer's value and the sellers' costs. Deng et al.~\cite{deng2025procurement} considered this problem in both online and offline settings, seeking to optimize the net surplus generated by the trade.

%% file: prelim.tex
\section{Preliminaries}\label{sec:model}

We consider settings with a principal who runs a platform and a set $N=\{1, \ldots, n\}$ of agents that can create content for this platform. If an agent $i\in N$ exerts effort $x_i\in [0, 1]$ to create content, they suffer a cost $x_i \cdot c_i$, where $c_i \in \mathbb{R}_{\ge 0}$ is a private parameter, and this content generates revenue $x_i \cdot r_i$, where $r_i\in \mathbb{R}_{\ge 0}$ is public. 
For convenience, we sometimes refer to an agent with $r_i \geq c_i$ as a
    \emph{funder} and to an agent with $r_i<c_i$ as a \emph{creator}; the label therefore depends on whether the revenue the agent can generate is at least their cost. 
    Note that since the cost of an agent is private information, the principal cannot distinguish between funders and creators.
Given an effort vector (or allocation vector) $\mathbf{x}=(x_i)_{i\in N}$ by the agents, the value of the content they create is denoted as $v(\mathbf{x})$, where $v: [0, 1]^n \to \mathbb{R}_{\ge 0}$ is an arbitrary non-decreasing subadditive function, i.e.,
for every $\mathbf{x},\mathbf{y}\in[0,1]^n$ such that
$\mathbf{x}+\mathbf{y}\leq\mathbf 1$ coordinatewise,
$v(\mathbf{x}+\mathbf{y}) \leq v(\mathbf{x})+v(\mathbf{y})$.
Canonical examples are linear valuation functions, $v(\mathbf{x}) = \sum_{i \in N} x_i \cdot v_i$, where $v_i$ denotes the marginal value contribution of agent $i$.

The goal is to design a mechanism that interacts with the agents to choose an allocation vector $\mathbf{x} \in [0,1]^n$ that determines the effort $x_i$ that each agent $i$ should exert, and a payment vector $\mathbf{p} \in \mathbb{R}^n$ that determines the payment $p_i$ that agent $i$ will receive in return. The objective of this (potentially randomized) mechanism is to maximize the expected value $\mathbb{E}[v(\mathbf{x})]$ of the allocation while satisfying (ex-post) \bb\ (BB) and
(ex-post) individual rationality (IR), i.e.,

\[\sum_{i \in N} p_i \le \sum_{i \in N} x_i \cdot r_i,\tag{BB}\]
\[p_i \geq c_i \cdot x_i \quad \text{for every agent } i \in N.\tag{IR}\]

Budget balance ensures that the total payment distributed to the agents does not exceed the revenue that they collectively generate, and individual rationality ensures that every agent's payment is at least the cost that the agent suffers.

\vspace{0.2cm}
\noindent\textbf{Types of Mechanisms.}
We consider two main types of mechanisms: simultaneous mechanisms in the form of direct revelation, and sequential mechanisms in the form of sequential reporting. In sealed-bid mechanisms (or direct-revelation mechanisms), each agent $i \in N$ is asked to report their private cost $c_i$ in the form of a bid $b_i \ge 0$. The mechanism $M = (\alloc, \pay)$ collects the reported bid profile $\mathbf{b} = (b_1, \ldots, b_n)$ as input and then returns the allocation vector $\alloc(\mathbf{b})$ and payment vector $\pay(\mathbf{b})$ as output. On the other hand, sequential protocols give rise to extensive-form games.

\vspace{0.2cm}
\noindent\textbf{Utility and Solution Concepts.}
Each agent $i\in N$ aims to maximize their utility $u_i(\mathbf{x}, \mathbf{p}) = p_i - c_i\cdot x_i$. We assume a standard tie-breaking rule such that if $p_i = c_i \cdot x_i$ (i.e., zero utility), agents prefer participation (a larger $x_i$).

We model the interaction between the principal and the agents as a Stackelberg game. The principal (leader) moves first by announcing and committing to a procurement mechanism. This mechanism induces a bidding game among the agents (followers), who choose strategies to maximize their own (expected) utility. We assume an information asymmetry between the principal and the agents: while the principal knows the revenue $r_i$ for each agent, they do not know the realized costs $\mathbf{c}$ of the agents. Consistent with the standard literature on the efficiency of auctions in equilibrium \cite{ckst16,roughgardenSyrgkanisTardos2017, GML25}, we model the interaction \emph{among} the agents as a game of complete information, where the cost vector $\mathbf{c}$ is common knowledge among the agents themselves.

\textbf{Simultaneous Games.} In one-shot bidding formats, each agent $i \in N$ submits a bid $b_i$ chosen from their strategy space $B_i = \R^+$. We denote the profile of all bids as $\mathbf{b} = (b_1, \ldots, b_n) \in \prod_{i \in N} B_i$.
A profile $\mathbf{b}$ is a Nash equilibrium (NE) if for every agent $i$, the bid $b_i$ maximizes their utility given the bids of others:
\[u_i(b_i, b_{-i}) \ge u_i(b'_i, b_{-i}) \quad \forall b'_i \in B_i.\]

In equilibrium, no agent can unilaterally change their bid to a different $b'_i$ to increase their payoff.
As a refinement of equilibrium in simultaneous direct-revelation mechanisms, a mechanism $M = (\alloc, \pay)$ is \emph{dominant-strategy incentive-compatible} (DSIC) if telling the truth is a dominant strategy. Formally, for every agent $i$ with true cost $c_i$, every possible deviation $b_i$, and any bidding profile $\mathbf{b}_{-i}$ of the other agents:
\[p_i(c_i, \mathbf{b}_{-i}) - c_i \cdot x_i(c_i, \mathbf{b}_{-i}) \geq  p_i(b_i, \mathbf{b}_{-i}) - c_i \cdot x_i(b_i, \mathbf{b}_{-i}).  \tag{DSIC}\]
A randomized mechanism is \emph{truthful-in-expectation} if the inequality above holds in expectation of the internal randomness of the mechanism:
\[\mathbb{E}[p_i(c_i, \mathbf{b}_{-i}) - c_i \cdot x_i(c_i, \mathbf{b}_{-i})] \geq \mathbb{E}[p_i(b_i, \mathbf{b}_{-i}) - c_i \cdot x_i(b_i, \mathbf{b}_{-i})].\]
\textbf{Sequential Games and Subgame Perfect Equilibria.} 
Consider an extensive-form game with a set of players $N$, represented by a rooted tree $\mathcal{T}$ of finite depth with node set $\mathcal{S}$. For any node $s \in \mathcal{S}$, we denote the set of its children as $\mathcal{T}(s)$.
\begin{itemize}
\item \textbf{Nodes:} If $\mathcal{T}(s) = \emptyset$, $s$ is a terminal node ($s \in \mathcal{S}_\text{term}$); otherwise, it is an internal node ($s \in \mathcal{S}_\text{int}$).
\item \textbf{Strategies:} Each internal node $s$ is associated with a player $i(s)$ whose turn it is to move. Let $\mathcal{S}_{i} \subseteq \mathcal{S}_\text{int}$ denote the set of nodes associated with agent $i$. A pure strategy for player $i$ is a mapping $a_i: \mathcal{S}_i \to \mathcal{S}$ such that $a_i(s) \in \mathcal{T}(s)$ (i.e., selecting a child node).
\item \textbf{Payoffs and Continuation Values:} Each terminal node $z \in \mathcal{S}_\text{term}$ is associated with a payoff $\pi_j(z) \in \mathbb{R}$ for each player $j$. Given a strategy profile $a = (a_1, \dots, a_n)$, we define the continuation payoff $\hat{\pi}_j(s \vert{} a)$ recursively:

$$\hat{\pi}_j(s \vert{} a) = \begin{cases} \pi_j(s) & \text{if } s \in \mathcal{S}_\text{term} \\ \hat{\pi}_j(a_{i(s)}(s) \mid a) & \text{if } s \in \mathcal{S}_\text{int} \end{cases}$$

\item \textbf{Subgame Perfect Equilibrium (SPE):} A strategy profile $a$ is an SPE if for every player $i$, at every node $s \in \mathcal{S}_i$, the strategy $a_i(s)$ maximizes the agent's continuation payoff. That is, for all possible alternative moves $s' \in \mathcal{T}(s)$:

$$\hat{\pi}_i(a_i(s) \mid a) \ge \hat{\pi}_i(s' \mid a).$$

\end{itemize}
Note that every SPE is a Nash equilibrium of the entire game.\footnote{Any extensive-form game can be mapped to a simultaneous normal-form game where strategies correspond to comprehensive plans of action for every possible scenario. In this mapping, an SPE strategy profile guarantees that no player can strictly benefit from unilaterally deviating to a different complete plan, thereby satisfying the definition of a Nash equilibrium.}

\vspace{0.2cm}
\noindent\textbf{Benchmarks.}
Given an instance $\mathcal{I} = (v, \revenue, \cost)$ and a set of agents $S \subseteq N$, let $\wel(S,\mathcal{I})$ denote the optimal feasible value {restricted to agents in $S$}, i.e., the maximum value achievable by a budget-feasible coalition formed entirely within $S$ under the valuation function $v(\cdot)$:

$$\mathcal{W}(S , \mathcal{I}) = \max_{\mathbf{x} \in [0,1]^{\vert{}S\vert{}}} \left\{ v(\mathbf{x}) \quad \bigg\vert{} \quad \sum_{i \in S} r_i \cdot x_i \ge \sum_{i \in S} c_i \cdot x_i \right\}.$$

We define the first-best benchmark to be $\opt(\mathcal{I}) = \wel(N,\mathcal{I})$, the maximum value achievable under known costs subject to the budget-balance constraint. Note that the budget-balance constraint of the benchmark is with respect to the \emph{costs}, rather than the payment.

%% file: approxopt.tex
\section{Inapproximability of the Optimal Value Benchmark}\label{sec:opt-lower-bounds}
In this section, we explore the landscape of approximating the optimal value ($\opt$) in autarkic marketplaces.
\cref{sec:dsic} focuses on the limitations of 
DSIC direct-revelation mechanisms, and \cref{sec:lowerbound} establishes an inapproximability  result when almost all information is public and there is a single strategic agent with a private cost.

\subsection{Unbounded Approximation of Truthful Direct-Revelation Mechanisms}\label{sec:dsic}

We begin by demonstrating that truthfulness and efficiency are incompatible in this setting: any DSIC mechanism incurs an unbounded approximation ratio to the optimal value, \opt, even for instances with just \emph{two agents} and a \emph{linear} valuation function,
where \emph{$\opt$ is known in advance}. Specifically, the agents used in our construction can generate revenue $r_1 = 1$ and $r_2 = 0$, respectively, and 
the valuation function is $v(x_1,x_2) = x_2$, so the first agent generates all the revenue (the ``funder'') and the other agent provides all the value (the ``creator''). The following proof shows that even if we know that it is feasible to hire both agents, i.e., $c_1+c_2\leq 1$ and $\opt=1$, the fact that the exact values of $c_1$ and $c_2$ are private makes it impossible for a DSIC mechanism to identify appropriate prices: the ``information rent'' that DSIC mechanisms need to pay to the agents on top of their true costs makes it impossible to satisfy budget-balance, i.e., $p_1+p_2\leq 1$, even though $c_1+c_2\leq 1$.

\begin{theorem}[Unbounded Approximation of DSIC Mechanisms]\label{thm:dsic-hardness-approximation}
   Let $\Delta$ be a set of two-agent instances with a linear valuation function for which it is known that $\opt=1$.
   Any DSIC, IR, and BB direct-revelation mechanism $\mathcal{M}$ has unbounded worst-case approximation to $\opt$ 
   over instances from $\Delta.$
\end{theorem}

\begin{proof}
    We consider a setting with two agents, where
    agent 1 is the funder with $r_1=1$ and agent 2 is the creator with  $r_2=0$. The value function is linear, i.e., $v(x_1,x_2) = v_1 \cdot x_1 + v_2\cdot x_2,$ where $v_1=0, v_2 =1.$
    Let $\Delta = \{(c_1, c_2) \in [0, 1]^2 \mid c_1 + c_2 < 1\}$
    be the set of the possible cost profiles.
    Within the region $\Delta$, the condition $c_1 + c_2 < 1$ ensures that selecting both agents at cost is budget feasible, so $\opt = 1$ for every profile in $\Delta$. 
    Assume for contradiction that a DSIC, IR, and BB mechanism guarantees a bounded approximation to $\opt$. Then, since $\opt =1$ across $\Delta$, there must exist $\alpha >0$ such that the allocation to the creator satisfies $x_2(c_1, c_2) \ge \alpha$ for all $(c_1, c_2) \in \Delta$.

    By Myerson's Lemma~\cite{myerson1981optimal}, the payment to agent 2 given cost $c_1$ is determined by integrating over their valid unilateral deviations up to the boundary of the type space:
    \[ p_2(c_1, c_2) = c_2 x_2(c_1, c_2) + \int_{c_2}^1 x_2(c_1, z) \, dz \]
    Evaluating this at $c_2 = 0$, and noting that $x_2 \ge \alpha$ whenever $z < 1 - c_1$ (the boundary of $\Delta$), we bound the payment:
    \[ p_2(c_1, 0) = \int_{0}^{1-c_1} x_2(c_1, z) \, dz \ge \int_{0}^{1-c_1} \alpha \, dz = \alpha(1 - c_1). \]
    Similarly, applying Myerson's Lemma to the funder at $c_2 = 0$ yields:
    \[ p_1(c_1, 0) = c_1 x_1(c_1, 0) + \int_{c_1}^1 x_1(z, 0) \, dz. \]
    The mechanism must satisfy (ex-post) BB: $p_1(c_1, 0) + p_2(c_1, 0) \le x_1(c_1, 0)$. Substituting our payment identities into this constraint, we obtain:
    \[ c_1 x_1(c_1, 0) + \int_{c_1}^1 x_1(z, 0) \, dz + \alpha(1 - c_1) \le x_1(c_1, 0). \]
    Let $f(c_1) = x_1(c_1, 0)$ and $F(c_1) = \int_{c_1}^1 f(z) \, dz$. By the Fundamental Theorem of Calculus, $F'(c_1) = -f(c_1)$. Grouping terms, we have:
    \[ F(c_1) + \alpha(1 - c_1) \le (1 - c_1)f(c_1) = -(1 - c_1)F'(c_1). \]
    Rearranging produces the following first-order linear differential inequality:
    \[ F'(c_1) + \frac{F(c_1)}{1 - c_1} \le -\alpha. \]
    Multiplying by the integrating factor $\frac{1}{1-c_1}$, we get:
    \[ \frac{d}{dc_1} \left[ \frac{F(c_1)}{1 - c_1} \right] \le -\frac{\alpha}{1 - c_1}. \]
    Integrating both sides from $c_1 = 0$ to an arbitrary $y \in (0, 1)$ gives:
    \[ \frac{F(y)}{1 - y} - F(0) \le \int_0^y -\frac{\alpha}{1 - z} \, dz = \alpha \ln(1 - y). \]
    Since $F(y) \ge 0$ and $1 - y > 0$, dropping the first term preserves the inequality:
    \[ -F(0) \le \alpha \ln(1 - y) \implies F(0) \ge -\alpha \ln(1 - y). \]
    As $y \to 1$, the right-hand side $-\alpha \ln(1 - y)$ diverges to $+\infty$. However, $F(0) = \int_{0}^1 x_1(z, 0) \, dz \le 1$. This contradiction implies that no DSIC mechanism can guarantee a bounded approximation to $\opt$ over $\Delta$. 
\end{proof}

\subsection{Impossibility Result for ``Single-Agent'' Instances}\label{sec:lowerbound}

\cref{thm:dsic-hardness-approximation} showed that even for instances with one funder agent, one creator agent, and a perfect estimate of $\opt$, it is impossible to guarantee any bounded approximation of $\opt$ using truthful mechanisms when both of the agents' costs are private. Note that if we also knew the cost $c_i$ of one of the two agents, we could guarantee the optimal outcome in that class of instances: we could offer a payment of $p_i=c_i$ to the agent whose cost we know and then use the remaining budget of $r_1+r_2-c_i$ to pay the other agent (guaranteeing budget-balance). Since we also know that $\opt=1$ for this class of instances, i.e., $c_1+c_2\leq 1$, this payment would satisfy individual rationality as well. However, this solution would heavily depend on the unrealistic assumption that we know the exact value of $\opt$ in advance. What if we are instead provided with an estimate in the form of a (normalized) range $[1, \sumval]$, within which $\opt$ lies? Here, a larger value of $\sumval$ corresponds to a less informative estimate of $\opt$.

Our next result considers this alternative estimate $[1, \sumval]$ and shows that for the same class of two-agent instances as in \cref{thm:dsic-hardness-approximation}, even if we let the exact cost of the creator be public and the only strategic agent that we interact with is the funder (whose cost remains private), it is impossible to achieve better than a logarithmic approximation as a function of $\sumval$. In other words, if we effectively further reduce this to a class of ``single-agent'' instances where the goal is to elicit, or approximately estimate, the cost of the funder, the performance deteriorates as a function of $\sumval$.

\begin{theorem}\label{thm:lowerbound-general}

Let $\Delta$ be a set of two-agent instances with a linear valuation function for which $\opt\in [1, \sumval]$ for some known $\sumval\geq 1$. Even if the cost of just one of the agents is private, for any (not necessarily direct revelation) IR and BB mechanism $\mathcal{M}$ that arbitrarily interacts with the other agent, there exists an instance in $\Delta$ where the approximation of $\mathcal{M}$ is $\Omega(\log \mathcal{V})$ in every equilibrium.
\end{theorem}

\begin{proof}

Consider an instance consisting of two agents. Agent 1 (the funder) is the only strategic agent, with a public revenue $r_1 = 1$, personal value $v_1 = 0$, and a private cost $c_1 \in [0, 1 - 1/\mathcal{V}]$. Agent 2 (the creator) has a public cost $c_2 = 1$, generates revenue $r_2 = 0$, and provides value $v_2 = \mathcal{V}$. Any profit ($r_1 - p$) extracted from the funder can be used to fractionally allocate the creator by setting $x_2 = 1 - p$, which results in a value of $x_2 \cdot \mathcal{V} = (1 - p)\mathcal{V}$.  

Because the funder is the only strategic agent, they face a single-player game against the mechanism $\mathcal{M}$. By the Revelation Principle \citep{DHM79, myerson1981optimal}, any equilibrium of this game can be mapped to a direct revelation mechanism where the funder truthfully reports $c_1$. For a single-parameter agent, the taxation principle shows that any such truthful mechanism is characterized by a monotone allocation rule, which is equivalent to a distribution over threshold prices $p$.  

The total revenue generated by the funder is fixed at $r_1 = 1$. We define a mapping between the posted price $p$ (the payment to the funder) and a surplus target $T$ (the value generated using the remaining budget):$$T(p) = (1 - p)\mathcal{V}.$$
Since $r_1 = 1$ is the only source of funding, $T$ represents the total value achievable in the system. The funder accepts the price $p$ if and only if $p \ge c_1$. This is equivalent to:
  $$1 - \frac{T}{\mathcal{V}} \ge c_1 \iff T \le (1 - c_1)\mathcal{V}$$
Let $S = (1 - c_1)\mathcal{V}$ be the true surplus value. The mechanism's choice of a distribution over prices $p$ is equivalent to a choice of a distribution over surplus targets $T \in [0, \mathcal{V}]$. 

Notice that in this instance $\opt = S$. To be $\alpha$-competitive, a randomized mechanism (defined by a distribution over surplus targets $T$) must satisfy:
$$\mathbb{E}[\text{Value} \mid S] = \int_{0}^{S} T \, d\mathbb{P}(T) \ge \frac{S}{\alpha} \quad \text{for every possible value of $S \in [1, \mathcal{V}]$}.$$

Consider a sequence of $k = \lfloor \log_2 \mathcal{V} \rfloor$ instances where the surplus is $S_j = 2^j$ for $j = 1, \dots, k$. We partition the target $T$ into the intervals $I_0 = [0,1]$ and $I_i = (2^{i-1},2^i]$ for $i=1,\dots,k$, and define $q_i = \mathbb{P}(T \in I_i)$ for $i=0,\dots,k$.

By the definition of $\alpha$-competitiveness, for any instance $j$, the expected value must satisfy:
\begin{equation}\label{eq:valueguarantee}
    q_0 + \sum_{i=1}^j q_i \cdot 2^i
    \geq
    \mathbb{E}[\text{Value} \mid S_j]
    \geq
    \frac{2^j}{\alpha}.
\end{equation}
Indeed, targets $T>2^j$ yield zero value in this instance, while $T\leq 1$ on $I_0$ and $T\leq 2^i$ on each $I_i$.

Dividing both sides of \cref{eq:valueguarantee} by $2^j$, we obtain:
\[
    q_0 \cdot 2^{-j}
    + \sum_{i=1}^j q_i \cdot 2^{i-j}
    \geq \frac{1}{\alpha}.
\]

Summing these inequalities over all $k$ instances $j=1,\dots,k$ yields:
\[
    q_0 \sum_{j=1}^k 2^{-j}
    + \sum_{j=1}^k \sum_{i=1}^j q_i \cdot 2^{i-j}
    \geq \frac{k}{\alpha}.
\]

Exchanging the order of summation in the second term gives:
\[
    q_0 \sum_{j=1}^k 2^{-j}
    + \sum_{i=1}^k q_i \sum_{j=i}^k 2^{i-j}
    \geq \frac{k}{\alpha}.
\]

Both inner sums are geometric series satisfying:
\[
    \sum_{j=1}^k 2^{-j} < 1
    \qquad\text{and}\qquad
    \sum_{j=i}^k 2^{i-j}
    = 1+\frac12+\cdots+\frac{1}{2^{k-i}} < 2. \qquad \Rightarrow \quad  \frac{k}{\alpha}
    <
    q_0 + 2\sum_{i=1}^k q_i.
\]

Since the intervals $I_0,I_1,\dots,I_k$ are disjoint,
\[
    q_0+\sum_{i=1}^k q_i \leq 1. \qquad \Rightarrow \qquad  q_0+2\sum_{i=1}^k q_i
    =
    2\left(q_0+\sum_{i=1}^k q_i\right)-q_0
    \leq 2.
\]

It follows that $\frac{k}{\alpha}<2$, and $\alpha>\frac{k}{2}$.

Substituting $k=\lfloor\log_2 \mathcal{V}\rfloor$, we conclude that:
\[
    \alpha>
    \frac{\lfloor\log_2 \mathcal{V}\rfloor}{2}
    =
    \Omega(\log \mathcal{V}).
    \qedhere
\]
\end{proof}

\begin{remark}[Alternative Parameterizations]\label{rem:alt-param}
It is natural to ask whether alternative parametrizations, such as dependence on the number of agents $n,$
can give more meaningful approximation guarantees.
The construction in \cref{thm:lowerbound-general} provides a justification for why $\sumval$ (the upper bound of the optimal value) is the appropriate parameter for analyzing the approximation ratio in this setting.
In particular, this instance rules out the possibility of achieving a bounded approximation with respect to other natural parameters commonly studied in auction theory:
\begin{itemize}
\item \textbf{Number of Agents ($n$):} The base construction relies on exactly two agents ($n=2$). Because the $\Omega(\log \mathcal{V})$ lower bound applies for a fixed $n$, the approximation ratio grows with the scale of the surplus $\mathcal{V}$, implying that no bounded approximation exists parameterized by $n$ alone. 
\item \textbf{Maximum Singleton Value ($v_{max}$):} If we adapt the instance by replacing the single creator (cost $1$, value $\mathcal{V}$) with $m$ identical creators, each with cost $1/m$ and marginal value $\mathcal{V}/m$, the underlying math and total maximum value $\mathcal{V}$ remain entirely unchanged. However, the maximum single-agent value becomes $v_{\max} = \mathcal{V}/m$. By making $m$ arbitrarily large, we can make $v_{\max}$ arbitrarily small while $\mathcal{V}$ remains large. Thus, the approximation ratio relative to the highest single-agent value, $\max_i(v_i)$, is unbounded.
\end{itemize}
\end{remark}

%% file: bafo_results.tex
\section{BAFO: Make Your Wish Come True}\label{sec:bafo}
In the previous section, we established that truthfulness imposes severe limitations on value, bounding efficiency by $O(\log \sumval)$, even with the knowledge that $\opt \in [1, \sumval]$. We now demonstrate that relaxing the truthfulness constraint and adopting a sequential, pay-as-bid format allows us to guarantee this logarithmic approximation in equilibrium. Our main result is the \guessbafo\ mechanism, which guarantees an $O(\log \sumval)$ approximation to the optimal value \opt\ in \emph{every} subgame perfect equilibrium (SPE).

\begin{theorem}\label{thm:guessbafo}
    The Guess-and-BAFO mechanism achieves an $O(\log \sumval)$ approximation to the optimal value \opt\ in every subgame perfect equilibrium (SPE) in the worst case.
\end{theorem}
The mechanism achieving this result is remarkably simple. It operates using two components. First, given the (normalized) range of $\opt \in [1,\sumval]$, it samples a target value threshold $T$ from a geometric scale over the range $[1, \sumval]$. Then, it uses a sequential BAFO protocol committed to a specific selection rule: {output some $\alloc \neq 0$ only if there exists $\alloc \in [0,1]^n$ such
that $\sum_i x_i \cdot b_i \leq \sum_i x_i \cdot r_i$ and $v(\alloc) \geq T.$\footnote{If multiple such $\alloc$ exist, it picks one according to a pre-defined tie-breaking rule, e.g.,
the maximum-value one.}
See 
\cref{alg:guess_bafo} for a formal description of the mechanism.}

\begin{algorithm}[ht]\caption{ \guessbafo~Mechanism}\label{alg:guess_bafo}
\SetAlgoLined\KwIn{agents $N=\{1, \dots, n\}$;  revenues $r_i$; upper bound of the optimal value $\sumval$.}
\KwOut{Allocation $\alloc^* \in [0,1]^n$, payments $\pay \in \mathbb{R}^n$.}
\BlankLine\tcp{Phase 1: The Guessing Phase}
Construct the set of powers of 2: $L \gets \{2^k \mid 2^k \leq \sumval, k \in \mathbb{N} \cup \{0\}\}$;

Uniformly select a target threshold $T$ from $L$ at random;

\BlankLine\tcp{Phase 2: Sequential BAFO Protocol}Announce the selection rule: ``Maximize $v(\alloc)$ subject to $\sum x_i b_i \leq \sum x_i r_i$ and $v(\alloc) \geq T$.'';

Approach agents sequentially (in an arbitrary order); each agent $i$ submits a BAFO bid $b_i$;

\BlankLine\tcp{Phase 3: Final Selection}Let $\mathbf{b}$ be the final bid vector. Compute the optimal valid set: $\alloc^* \in \arg\max_x \{\sum x_i v_i \mid \left(\sum x_i b_i \leq \sum x_i r_i\right) \land \left(v(\alloc) \geq T\right)\}$;

\If{no such $\alloc^*$ exists}{Output $\alloc = \mathbf{0}, \pay = \mathbf{0}$;}

\Else{Allocate to $\alloc^*$ and pay $p_i = b_i \cdot x_i$ for all winners;}
\end{algorithm}

The effectiveness of this approach relies on a fundamental property of the sequential BAFO protocol which we term the \emph{wish-come-true} capability. We show that if the principal defines a selection rule based on \emph{any} price-monotonic condition (such as a value threshold) that is feasible under the agents' true costs, the sequential structure ensures that agents will coordinate to satisfy this condition in \emph{every} SPE. This general property allows the principal to enforce the conjectured optimal value, effectively reducing the value maximization problem to a guessing problem.

In stark contrast to DSIC mechanisms—where the approximation ratio remains unbounded even if the optimal value $\opt$ is known—the sequential BAFO protocol allows a principal with this knowledge to achieve \opt\ in every SPE.
\begin{corollary}\label{cor:withprediction}
If the principal knows the optimal value $\opt$, then the Sequential BAFO mechanism employing the selection rule
\[
\alloc \in \arg\max_{\mathbf{x}} \bigg\{ v(\alloc) \bigm| \left(\sum x_i b_i \le \sum x_i r_i\right) \land \left(v(\alloc) \ge \opt\right) \bigg\}
\]
achieves value exactly $\opt$ in every subgame perfect equilibrium.
\end{corollary}

In \cref{sec:bafoprotocol}, we formally define the sequential BAFO protocol and establish the general wish-come-true result. In \cref{sec:approxvaluebafo}, we instantiate this framework with the \guessbafo\ algorithm and prove the logarithmic approximation guarantee.

\subsection{BAFO as a Constraint Solver}\label{sec:bafoprotocol} The BAFO interaction proceeds in $n$ rounds\footnote{We identify agent $i$ with round $i$; this is only for notational convenience and can be done by relabeling the agents.}. In every round $i$, 
the principal asks the $i$-th agent for their bid, and the agent submits it knowing all bids reported by agents $j <i$. 
In round $i$, agent $i$ observes the history of bids $b_{<i} = (b_1, \dots, b_{i-1})$ and submits a best and final offer $b_i \ge 0$. This bid represents a binding commitment to sell at price $b_{i}$ if selected.\footnote{We note that the price is a per-unit price, i.e., if the allocation is $x_i \leq1$, the price we pay to the agent is $x_i b_i$.} After all bids $b = (b_{1}, \dots, b_{n})$ are collected, the principal selects a subset of winners $\alloc \in [0,1]^n$ according to some public selection rule $f(\bids)$ and pays winning agents $p_{i} = b_{i} \cdot x_{i}$.

The BAFO procedure was originally analyzed by Gkatzelis et al.~\cite{GML25} for maximizing the principal's utility in procurement mechanisms, where the valuation function belongs to a broad class of combinatorial functions.
We identify a fundamental capability of the protocol that we call ``wish-come-true'': if a specific market condition is feasible at the agents' true costs, the sequential structure ensures that agents coordinate to satisfy that condition in every SPE. We think this property of sequential BAFO is of independent interest, given the generality of conditions the protocol can enforce. 

To formalize this, we model the principal's requirement as a Boolean predicate $\Phi(\alloc, \mathbf{b})$. We identify two key conditions—one on the predicate and one on the selection rule—that are sufficient to guarantee this property.

\begin{definition}[Price Monotonic Predicate]\label{def:monotone}A predicate $\Phi(\alloc, \mathbf{b})$ is price monotonic if for any allocation $\alloc$ and any two bid vectors $\mathbf{b}, \mathbf{b}'$:$$\left( \Phi(\alloc, \mathbf{b}) = \text{True}  \land \forall  i\  s.t.\ x_i >0,\ b'_i \le b_i \right) \implies \Phi(\alloc, \mathbf{b}') = \text{True}.$$\end{definition}
In the context of the \guessbafo\ mechanism, the predicate $\Phi$ is the conjunction of the budget constraint and the target value constraint. As we show later, such a predicate is price monotonic.

Next, we specify the conditions the principal's selection rule and the predicate need to satisfy. The principal need not necessarily choose the value-maximizing outcome,
provided they commit to picking \emph{some} valid outcome if one exists, and do not allocate to any agent if 
no such outcome exists.
\begin{definition}[$\Phi$-consistent Selection Rule]\label{def:consistent_rule}A selection rule $f(\mathbf{b})$ is $\Phi$-consistent if it satisfies:
\begin{enumerate}
\item \textbf{Feasibility Compliance:} If the set of valid allocations $\mathcal{X}_\Phi(\mathbf{b}) = \{\alloc \mid \Phi(\alloc, \mathbf{b}) = \text{True}\}$ is non-empty, then $f(\mathbf{b}) \in \mathcal{X}_\Phi(\mathbf{b})$. On the other hand, if the set of valid allocations $\mathcal{X}_\Phi(\mathbf{b}) = \{\alloc \mid \Phi(\alloc, \mathbf{b}) = \text{True}\}$ is empty, then $f(\mathbf{b}) = (0,\ldots,0).$
\item \textbf{Deterministic Tie-Breaking:} If multiple allocations satisfy $\Phi$, the principal selects one according to a public, deterministic ordering. 
\end{enumerate}
\end{definition}

With these conditions met, the protocol guarantees success:

\begin{theorem}\label{thm:implement}
Consider a sequential BAFO procedure with a price-monotonic predicate $\Phi$ and a $\Phi$-consistent selection rule $f$.
If there exists any allocation $\alloc^*$ that is feasible at the agents' true costs (i.e., $\Phi(\alloc^*,\mathbf{c})=\textnormal{True}$),
then in every SPE induced by $f$, the terminal allocation $\alloc$ satisfies
$\Phi(\alloc,\mathbf{b})=\textnormal{True}$.
\end{theorem}

 \noindent \textbf{Remark (Feasibility vs. Optimality).} \cref{thm:implement} guarantees that \emph{a} valid solution will be found if one exists, but not necessarily the \emph{optimal} one. Agents might coordinate on an SPE that satisfies $\Phi$ but yields lower value than another valid outcome if that specific equilibrium offers them higher payoffs. This is why the \guessbafo\ mechanism must ``guess'' a target $T$: by embedding the efficiency requirement directly into the feasibility predicate $\Phi$ (via $v(\alloc) \ge T$), we force the agents to coordinate on a high-value outcome to satisfy the constraint.

\begin{proof}
Fix an arbitrary pure SPE $\sigma$.\footnote{The argument extends to mixed
SPE by applying the induction to every action in the support at every
subgame. In the second case below, an action that destroys conditional
feasibility yields zero allocation, whereas bidding one's cost yields zero
utility with positive allocation. Hence, such an action cannot belong to the
support.}
Let $s$ be any node at level $k$ (so agent $k$ moves next), with past bids
$b_{<k}(s)$. Define the conditional price vector
\[
\hat p_i(s)=
\begin{cases}
b_i(s) & \text{if } i<k,\\
c_i & \text{if } i\ge k.
\end{cases}
\]
Let $\mathbf{b}(s)$ be the terminal bid vector generated by $\sigma$ from
$s$, and let $\alloc(s)=f(\mathbf{b}(s))$ be the resulting terminal
allocation.

We prove by backward induction on $k=n,n-1,\ldots,1$ that, for every node
$s$ at level $k$,
\begin{equation}\label{eq:goal}
\mathcal X_\Phi\bigl(\hat{\mathbf p}(s)\bigr)\neq\emptyset
\quad\Longrightarrow\quad
\mathcal X_\Phi\bigl(\mathbf{b}(s)\bigr)\neq\emptyset.
\end{equation}

We first record a simple observation that will be used throughout the proof.
At any node where agent $i$ moves, bidding $b_i=c_i$ guarantees
utility
$
(b_i-c_i)x_i=0
$
regardless of the terminal allocation. Therefore, in an SPE continuation,
any agent who is selected with positive allocation must bid at least their
cost. 
In particular, for any node $s$ at level $k$,
\begin{equation}\label{eq:nonneg}
i\ge k
\quad\text{and}\quad
\alloc_i(s)>0
\quad\Longrightarrow\quad
b_i(s)\ge c_i.
\end{equation}

\paragraph{Base case ($k=n$).}
Fix a node $s$ at level $n$ and suppose that
\[
\mathcal X_\Phi\bigl(\hat{\mathbf p}(s)\bigr)
=
\mathcal X_\Phi\bigl(b_{<n}(s),c_n\bigr)
\neq\emptyset.
\]

First suppose that there exists
$y\in\mathcal X_\Phi(\hat{\mathbf p}(s))$ with $y_n=0$.
Let $b_n$ be agent $n$'s equilibrium bid. Since changing the bid of an
agent with zero allocation does not change any bid on the support of $y$,
price monotonicity implies
$
\Phi\bigl(y,(b_{<n}(s),b_n)\bigr)=\textnormal{True}.
$
Hence,
$
\mathcal X_\Phi\bigl(\mathbf b(s)\bigr)\neq\emptyset.
$

Now suppose that every allocation in
$\mathcal X_\Phi(\hat{\mathbf p}(s))$ assigns agent $n$ a positive
allocation. Consider the deviation in which agent $n$ bids $c_n$. Under
this deviation, the terminal price vector is exactly
$\hat{\mathbf p}(s)$. Since this feasible set is nonempty,
$\Phi$-consistency implies that the resulting allocation is feasible.
Moreover, by assumption, it assigns agent $n$ a positive allocation.
Thus, bidding $c_n$ gives agent $n$ zero utility with positive
participation.

If agent $n$'s equilibrium bid instead resulted in
$
\mathcal X_\Phi\bigl(\mathbf b(s)\bigr)=\emptyset,
$
then $\Phi$-consistency would imply that the terminal allocation is
$\mathbf 0$. Agent $n$ would therefore receive zero utility with no
participation, which is strictly worse, under the tie-breaking rule, than
bidding $c_n$ and receiving zero utility with positive participation.
This contradicts optimality of the equilibrium action. Therefore,
\[
\mathcal X_\Phi\bigl(\mathbf b(s)\bigr)\neq\emptyset,
\]
proving \eqref{eq:goal} for level $n$.

\paragraph{Inductive step.}
Assume that \eqref{eq:goal} holds for all nodes at levels
$k+1,\ldots,n$. Fix a node $s$ at level $k$ and suppose that
\[
\mathcal X_\Phi\bigl(\hat{\mathbf p}(s)\bigr)\neq\emptyset,
\qquad
\hat{\mathbf p}(s)
=
\bigl(b_{<k}(s),c_k,c_{>k}\bigr).
\]

We distinguish two cases.

\medskip
\noindent
\emph{Case 1: There exists a feasible allocation that does not use agent
$k$.}

Suppose that there exists
$y\in\mathcal X_\Phi(\hat{\mathbf p}(s))$ with $y_k=0$.
Let $b_k$ be agent $k$'s equilibrium bid at $s$, and let $s'$ be the
child reached after this bid. Then
$
\hat{\mathbf p}(s')
=
\bigl(b_{<k}(s),b_k,c_{>k}\bigr).
$
The two conditional price vectors differ only in coordinate $k$, and
$y_k=0$. Hence, all prices on the support of $y$ are unchanged. By price
monotonicity,
$
\Phi\bigl(y,\hat{\mathbf p}(s')\bigr)=\textnormal{True}.
$
Therefore,
$
\mathcal X_\Phi\bigl(\hat{\mathbf p}(s')\bigr)\neq\emptyset.
$
Applying the induction hypothesis at $s'$ gives
$
\mathcal X_\Phi\bigl(\mathbf b(s')\bigr)\neq\emptyset.
$
Since $\mathbf b(s')=\mathbf b(s)$, this proves \eqref{eq:goal} in this
case.

\medskip
\noindent
\emph{Case 2: Every feasible allocation uses agent $k$.}

Suppose that every
$y\in\mathcal X_\Phi(\hat{\mathbf p}(s))$ satisfies $y_k>0$.
Consider the deviation in which agent $k$ bids $c_k$, and let $s^c$ be
the corresponding child. Since
$
\hat{\mathbf p}(s^c)
=
\bigl(b_{<k}(s),c_k,c_{>k}\bigr)
=
\hat{\mathbf p}(s),
$
we have
\[
\mathcal X_\Phi\bigl(\hat{\mathbf p}(s^c)\bigr)\neq\emptyset.
\]
By the induction hypothesis,
$
\mathcal X_\Phi\bigl(\mathbf b(s^c)\bigr)\neq\emptyset.
$

Let
$
\alloc^c=f\bigl(\mathbf b(s^c)\bigr)
$
be the terminal allocation following this deviation. By
$\Phi$-consistency,
\[
\Phi\bigl(\alloc^c,\mathbf b(s^c)\bigr)=\textnormal{True}.
\]

We claim that $\alloc^c_k>0$. Suppose otherwise that
$\alloc^c_k=0$. Every future agent $i>k$ with
$\alloc^c_i>0$ bids at least their cost by \eqref{eq:nonneg}, because
the restriction of an SPE to the subgame following $s^c$ is itself an
SPE. Starting from the terminal vector $\mathbf b(s^c)$, lower the bid
of each such future winner to $c_i$. The bids of agents $i<k$ remain
unchanged, agent $k$ already bids $c_k$, and changes to the bids of
agents receiving zero allocation are irrelevant. Price monotonicity
therefore implies
\[
\Phi\bigl(\alloc^c,\hat{\mathbf p}(s)\bigr)
=
\textnormal{True}.
\]
But $\alloc^c_k=0$, contradicting the assumption that every feasible
allocation at $\hat{\mathbf p}(s)$ uses agent $k$. Hence,
$
\alloc^c_k>0.
$
Thus, by bidding $c_k$, agent $k$ can guarantee zero utility with
positive participation.

Now let $b_k$ be agent $k$'s equilibrium bid at $s$, and let $s'$ be
the child reached under the equilibrium strategy. We claim that
$
\mathcal X_\Phi\bigl(\hat{\mathbf p}(s')\bigr)\neq\emptyset.
$
Suppose, toward a contradiction, that
$
\mathcal X_\Phi\bigl(\hat{\mathbf p}(s')\bigr)=\emptyset.
$
We show that this also implies
\[
\mathcal X_\Phi\bigl(\mathbf b(s')\bigr)=\emptyset.
\]
Indeed, suppose instead that the latter set were nonempty. Then, by
$\Phi$-consistency, the terminal allocation
$
\alloc(s')=f\bigl(\mathbf b(s')\bigr)
$
would satisfy
$
\Phi\bigl(\alloc(s'),\mathbf b(s')\bigr)=\textnormal{True}.
$
For every future winner $i>k$, \eqref{eq:nonneg} gives
$b_i(s')\ge c_i$. Lowering the bids of all such future winners to their
true costs and applying price monotonicity would then give
\[
\Phi\bigl(\alloc(s'),\hat{\mathbf p}(s')\bigr)
=
\textnormal{True},
\]
contradicting
$\mathcal X_\Phi(\hat{\mathbf p}(s'))=\emptyset$.
Consequently,
$\
\mathcal X_\Phi\bigl(\mathbf b(s')\bigr)=\emptyset,
$
and $\Phi$-consistency implies that the terminal allocation is
$\mathbf 0$. Agent $k$ would therefore receive zero utility with no
participation. This is strictly worse, under the tie-breaking rule, than
the deviation $b_k=c_k$, which gives zero utility with
$\alloc^c_k>0$. This contradicts the assumption that $b_k$ is an
equilibrium action.

It follows that
$
\mathcal X_\Phi\bigl(\hat{\mathbf p}(s')\bigr)\neq\emptyset.
$
Applying the induction hypothesis at $s'$ gives
$
\mathcal X_\Phi\bigl(\mathbf b(s')\bigr)\neq\emptyset.
$
Since $\mathbf b(s')=\mathbf b(s)$, this proves \eqref{eq:goal}.

This completes the backward induction. At the root $s_{\mathrm{root}}$,
we have
$
\hat{\mathbf p}(s_{\mathrm{root}})=\mathbf c.
$
By the premise of the theorem,
$\mathcal X_\Phi(\mathbf c)\neq\emptyset$, so \eqref{eq:goal} implies
that the terminal feasible set is nonempty:
$
\mathcal X_\Phi(\mathbf b)\neq\emptyset.
$
Finally, $\Phi$-consistency of $f$ implies
$
f(\mathbf b)\in\mathcal X_\Phi(\mathbf b),
$
and hence the terminal allocation $\alloc=f(\mathbf b)$ satisfies
\[
\Phi(\alloc,\mathbf b)=\textnormal{True}.
\qedhere
\]
\end{proof}

\subsection{Approximating Optimal Value via Guessing}\label{sec:approxvaluebafo}

We now leverage the wish-come-true property (\cref{thm:implement}) to approximate the optimal value. The strategy is to augment the standard budget constraint with a target value threshold $T$, restricting the feasible set to allocations that are both budget-balanced and satisfy $v(\alloc) \ge T$. Since this combined predicate is price-monotonic, \cref{thm:implement} guarantees that if a chosen threshold $T$ is feasible (i.e., $T \le \opt$), the agents will coordinate to satisfy it in \emph{every} SPE.

Consequently, the mechanism designer's challenge reduces to selecting a threshold $T$ that is sufficiently ambitious, yet feasible. We refer to the specific instantiation of the protocol used in Phase 2 of \cref{alg:guess_bafo} with a fixed target $T$ as \emph{Targeted-BAFO}. We analyze the performance of this approach in two settings: one where the principal has an \emph{educated guess} of the optimal value, and one where they do not.

\paragraph{Learning-Augmented BAFO.}
We first consider the setting where the principal has access to a prediction $\pred$ with a known error bound $\bar{\eta} \ge 1$ such that $\opt/\bar{\eta} \le \pred \le \bar{\eta} \cdot \opt$. In this case, the principal can set a deterministic threshold derived from the prediction. The proof is deferred to \cref{apx:bafo}.

\begin{theorem}\label{thm:learning_bafo}
     Targeted-BAFO with $T = \pred / \bar{\eta}$ achieves a value of at least $\frac{\opt}{\bar{\eta}^2}$ in every SPE.
\end{theorem}

\vspace{-0.5cm}
\paragraph{The \guessbafo\ Mechanism.}
Without {targeted estimates of \opt}, the principal cannot pinpoint a safe optimal threshold. To address this, the \guessbafo\ mechanism employs a randomized search over a geometric scale of possible values $L = \{2^k \mid 2^k \leq \sumval\}$. This effectively randomizes the choice of $T$ passed to the Targeted-BAFO subroutine.

\begin{proof}[Proof of \cref{thm:guessbafo}]
    Let $\opt$ denote the maximum value achievable by any budget-balanced fractional allocation with respect to the true costs $c$. That is,$$\opt = \max_{x \in [0,1]^n} \left\{v(x) \mid \sum x_i c_i \le \sum x_i r_i\right\}.$$
    Consider the threshold $T^* \in L$ such that $\opt/2 < T^* \le \opt$. Since the set $L$ contains powers of 2 up to $\sumval$, the size of the relevant search space is $|L| = \Theta(\log \sumval)$. The probability that the mechanism selects exactly this threshold is $1/|L| = \Omega(1/\log \sumval)$.

    Suppose the mechanism selects $T^*$. 
    We analyze the constraints imposed in Phase 3. 
    The buyer requires an allocation $\alloc$ satisfying the following conditions:
    budget-balance: $\sum x_i b_i \le \sum x_i r_i$ and minimum value: $v(\alloc) \ge T^*$.
    Let $\Phi(x, b)$ be the conjunction of these two conditions. Note that $\Phi$ is \emph{price monotonic} according to \cref{def:monotone}: if an allocation is budget-balanced at prices $b$, it strictly satisfies it at lower prices $b' \le b$, and the value constraint depends only on the allocation $\alloc$, not the prices.
    Since $T^* \le \opt$, we know by definition that there exists at least one allocation (the optimal set) that is feasible with respect to true costs $c$ and satisfies the value threshold. Therefore, by \cref{thm:implement}, the sequential BAFO protocol acts as a feasibility solver and will guarantee an outcome $x_{out}$ such that $\Phi(x_{out}, b)$ is True.
    Consequently, whenever $T^*$ is selected, the realized value is at least $T^* > \opt/2$. The expected value is therefore:
    \[\mathbb{E}\left[v(\alloc)\right] \ge \Pr\left[T=T^*\right] \cdot \frac{\opt}{2} = \Omega\left(\frac{\opt}{\log \sumval}\right). \qedhere\]
\end{proof}

\paragraph{Coordination Failure in Simultaneous Protocols.}For our positive results, we relied on the {wish-come-true} property of the sequential protocol. A natural question is whether this property holds for simultaneous BAFO protocols. We show that it does not: simultaneous mechanisms suffer from bad equilibria due to coordination failure.

Consider an instance with two agents, a linear valuation $v(\alloc) = 0 \cdot x_1 + 2 \cdot x_2$, and revenue-cost profiles $(r_1, c_1) = (1, 0.5)$ and $(r_2, c_2) = (0, 0.5)$. The optimal solution selects both agents, generating a value of $\opt = 1$. Suppose the designer knows $\opt=1$ and commits to a rule: ``Select a budget-balanced allocation with value at least $1$.'' In a simultaneous bidding game, the profile $(b_1, b_2) = (1, 1)$ constitutes a Nash equilibrium where no trade occurs, since there is no valid strategy that either agent can deviate to in order to trigger the allocation.

%% file: mms.tex
\section{Approximating the \mms\ Benchmark}\label{sec:mms}
In  \cref{sec:opt-lower-bounds}, we demonstrated that the optimal budget-balanced value, \opt, is an {overly} ambitious benchmark for self-funded markets. We proved that no DSIC mechanism can achieve a non-trivial approximation, and even under relaxed equilibrium concepts, efficiency is bounded by $\Omega(\log \sumval)$, where $\sumval$ is the upper bound of $\opt$. These negative results rely on ``thin'' market instances, where revenue generation depends entirely on a single agent. In such quasi-monopolistic settings, the lack of competition allows pivotal agents to exert excessive pricing power, making the extraction of their surplus impossible.

This motivates the search for a benchmark that accurately captures the \emph{competitiveness} of an instance. {Due to the nature of our problem, it is in “thick” markets that we can expect mechanisms to be able to generate value.}

\paragraph{The Challenge: Classifying Heterogeneous Agents.}
Established metrics for market thickness (e.g., \emph{frugal solution} or \emph{$\mathcal{F}^{(2)}$)} are not sufficient in our setting because agent contributions are two-fold: they provide both \emph{value} and \emph{revenue}. In a budget-balanced ecosystem, a high-value agent might only be feasible due to the cross-subsidy provided by a specific revenue-generating agent. Consequently, one cannot simply classify agents into ``types'' (e.g., value generators vs. revenue generators) to check for competition within each class; the complex interdependence of costs, revenues, and values makes individual-level classification intractable. See \cref{apx:benchmark} for a more detailed discussion of different benchmarks in our setting.

To circumvent these issues, we shift our focus from individual agents to \emph{aggregate sub-markets}. Since we cannot easily determine whether a specific agent faces competition, we instead test whether the market \emph{as a whole} faces competition from within itself. We ask: is the market robust enough to be split into two disjoint, viable sub-economies?

We formalize this using the \emph{Maximin Share} (\mms) benchmark, which is a well-studied notion in fair division. We adapt \mms\ here as a means to test for market competition. Intuitively, this benchmark creates a hypothetical scenario where the market is partitioned into two competing sub-markets (in the best way possible). If the original market is truly competitive, both sub-markets should be able to create meaningful value. On the other hand, if the market relies on a single agent (as in our lower bounds), one of the sub-markets will inevitably fail, collapsing the benchmark to zero.

\begin{definition}[Maximin Share (MMS)]
Given an instance $\mathcal{I} = (v, \revenue, \cost)$, the Maximin Share is the maximum value of the worst sub-market across all possible bipartitions of the agents:
\[
\mms(\mathcal{I}) = \max_{(P_1, P_2) \in \Pi_2(N)} \min \Big\{ \mathcal{W}(P_1,\mathcal{I}), \mathcal{W}(P_2,\mathcal{I}) \Big\},
\]
where $\Pi_2(N)$ denotes the set of all disjoint partitions $P_1 \cup P_2 = N$, and $\mathcal{W}(S,\mathcal{I})$ denotes the optimal budget-balanced value achievable when restricted to the sub-market $S$.
\end{definition}

This benchmark effectively filters out the problematic instances from \cref{sec:opt-lower-bounds}. In those cases, the single revenue provider must belong to either $P_1$ or $P_2$; the competing sub-market without them has no budget, causing the minimum value to be zero. By contrast, in a large, redundant market (e.g., one that is a union of two identical economies), the sub-markets can compete on equal footing, and the MMS captures a constant fraction of the optimal value ($\mms \approx \opt/2$).

 We now aim to approximate the \mms\ benchmark. 
We first show that truthfulness remains a barrier: even against this relaxed benchmark, no truthful mechanism can achieve a constant approximation. In contrast, we propose a new mechanism, the \emph{\ladderbafo}, and prove that it achieves a constant-factor approximation to the \mms\ benchmark in every SPE.

\subsection{Lower Bound for Truthful Mechanisms}

As a first step towards understanding the approximability of $\mms$ by general mechanisms, we ask
whether DSIC mechanisms can achieve a constant-factor approximation to it for every instance
$\mathcal{I}.$ \cref{thm:mms-logn-lower} gives a negative answer, showing that any DSIC mechanism
must suffer a $\log(n)$ degradation, where $n$ is the number of agents. 

The intuition stems from the fundamental gap between the \emph{available} budget in the system and the \emph{extractable} budget under truthful mechanisms.
In our setting, ``funders'' generate slack defined by the difference between their fixed revenue and their private cost ($s_i = r_i - c_i$). To pay for the expensive creators, the mechanism must capture this slack.
While the MMS benchmark can utilize the \emph{total} available slack (since it knows true costs), a truthful mechanism is limited by information asymmetry. In particular, we construct distributions in which any truthful mechanism, unable to price-discriminate perfectly, can only capture a small fraction of it ($1/\log n$).

\begin{theorem}\label{thm:mms-logn-lower}
   Consider any randomized direct-revelation mechanism $\mathcal{M}$ that is
truthful in expectation, IR in expectation, and ex-post BB. Then, the worst-case approximation of $\mathcal{M}$ to the \mms\ benchmark is $\Omega(\log n)$.
\end{theorem}

\begin{proof}
Let \(n=2m+2\), where \(m\geq 2\), and write
$
    H_m:=\sum_{k=1}^m \frac{1}{k}.
$
Consider two disjoint sets \(F_1,F_2\) of \(m\) funders and two
creators \(a_1,a_2\). For every funder \(i\in F_1\cup F_2\), set
\(r_i=1\), independently draw
\[
    s_i
    \sim
    \operatorname{Unif}
    \left\{
        1,\frac12,\ldots,\frac1m
    \right\},
\]
and set \(c_i=1-s_i\). For the creators, set
$
    r_{a_1}=r_{a_2}=0$,
    and $
    c_{a_1}=c_{a_2}=\frac{H_m}{2},
$
and let
$
    v(\alloc)=x_{a_1}+x_{a_2}.
$
Let \(\mathcal D\) denote the resulting product distribution over
cost profiles.

For each \(\ell\in\{1,2\}\), let
\[
    Z_\ell:=\sum_{i\in F_\ell}s_i.
\]
If \(Z_\ell\geq H_m/2\), then fully allocating every agent in
\(F_\ell\cup\{a_\ell\}\) is budget feasible, since
\[
    \sum_{i\in F_\ell\cup\{a_\ell\}}(r_i-c_i)
    =
    Z_\ell-\frac{H_m}{2}
    \geq 0.
\]
This allocation has value \(1\). Hence, on the event
\[
    E
    :=
    \left\{Z_1\geq\frac{H_m}{2}\right\}
    \cap
    \left\{Z_2\geq\frac{H_m}{2}\right\},
\]
the partition
$
    \bigl(F_1\cup\{a_1\},\,F_2\cup\{a_2\}\bigr)
$
witnesses that the \(\mms\) benchmark of the realized instance is at
least \(1\).

We next show that \(E\) has constant probability. For every funder,
\[
    \E[s_i]=\frac{H_m}{m},
    \qquad
    \E[s_i^2]
    =
    \frac1m\sum_{k=1}^m\frac1{k^2}
    \leq
    \frac{\pi^2}{6m}.
\]
Consequently, for each \(\ell\in\{1,2\}\),
\[
    \E[Z_\ell]=H_m,
    \qquad
    \E[Z_\ell^2]
    =
    \E[Z_\ell]^2+\operatorname{Var}(Z_\ell)
    \leq
    H_m^2+\frac{\pi^2}{6}.
\]
The Paley--Zygmund inequality~\cite{PZ32} gives
\[
    \Pr\left[Z_\ell\geq\frac{H_m}{2}\right]
    \geq
    \frac{H_m^2}
         {4\left(H_m^2+\pi^2/6\right)}
    \geq
    \frac{H_2^2}
         {4\left(H_2^2+\pi^2/6\right)}
    >0.
\]
Since \(Z_1\) and \(Z_2\) are independent,
\(\Pr_{\cost\sim\mathcal D}[E]\) is bounded below by a universal
positive constant.

For every report profile \(\bids\), define
$
    \bar x_i(\bids)
    :=
    \E_{\mathcal M}[x_i(\bids)],
$ and $
    \bar p_i(\bids)
    :=
    \E_{\mathcal M}[p_i(\bids)],
$
where the expectations are over the internal randomization of
\(\mathcal M\). At every truthful cost profile \(\cost\), IR in
expectation for the two creators and the expectation of the ex-post
BB constraint imply
\begin{align}
    \frac{H_m}{2}
    \E_{\mathcal M}\bigl[v(\alloc(\cost))\bigr]
    &=
    \frac{H_m}{2}
    \bigl(
        \bar x_{a_1}(\cost)+\bar x_{a_2}(\cost)
    \bigr)
    \notag\\
    &\leq
    \bar p_{a_1}(\cost)+\bar p_{a_2}(\cost)
    \notag\\
    &\leq
    \sum_{i\in F_1\cup F_2}
    \bigl(
        \bar x_i(\cost)-\bar p_i(\cost)
    \bigr).
    \label{eq:mms-value-from-funder-slack}
\end{align}

We now bound the expected contribution of a single funder to the last
sum. Fix \(i\in F_1\cup F_2\), condition on \(\cost_{-i}\), and
suppress \(\cost_{-i}\) from the notation. Truthfulness in expectation
implies that \(\bar x_i(z)\) is non-increasing in \(z\), and the
single-parameter payment identity gives
\[
    \bar p_i(z)
    =
    z\bar x_i(z)
    +
    \int_z^1\bar x_i(y)\,dy
    +
    \bigl(
        \bar p_i(1)-\bar x_i(1)
    \bigr).
\]
The final term is nonnegative by IR in expectation at type \(1\).
Therefore, for every \(s\in[0,1]\),
\begin{equation}
    \bar x_i(1-s)-\bar p_i(1-s)
    \leq
    s\bar x_i(1-s)
    -
    \int_0^s\bar x_i(1-t)\,dt.
    \label{eq:mms-pointwise-funder-slack}
\end{equation}

Because \(t\mapsto\bar x_i(1-t)\) is non-decreasing, for every
\(k\in[m]\),
\[
    \int_0^{1/k}\bar x_i(1-t)\,dt
    \geq
    \sum_{j=k}^{m-1}
    \left(
        \frac1j-\frac1{j+1}
    \right)
    \bar x_i\left(1-\frac1{j+1}\right).
\]
Averaging \cref{eq:mms-pointwise-funder-slack} over the \(m\) possible
costs of funder \(i\), and exchanging the two sums gives
\begin{align*}
    \E\left[
        \bar x_i(c_i)-\bar p_i(c_i)
        \,\middle|\,
        \cost_{-i}
    \right]
    \leq
    \frac1m
    \left[
        \sum_{k=1}^m
        \frac1k
        \bar x_i\left(1-\frac1k\right)
        -
        \sum_{j=1}^{m-1}
        \frac1{j+1}
        \bar x_i\left(1-\frac1{j+1}\right)
    \right]
    =
    \frac{\bar x_i(0)}{m}
    \leq
    \frac1m.
\end{align*}
Averaging over \(\cost_{-i}\) and summing over the \(2m\) funders
therefore gives
\[
    \E_{\cost\sim\mathcal D}
    \left[
        \sum_{i\in F_1\cup F_2}
        \bigl(
            \bar x_i(\cost)-\bar p_i(\cost)
        \bigr)
    \right]
    \leq 2.
\]
Taking expectations in
\cref{eq:mms-value-from-funder-slack}, we conclude that
\[
    \E_{\cost\sim\mathcal D}
    \E_{\mathcal M}
    \bigl[v(\alloc(\cost))\bigr]
    \leq
    \frac{4}{H_m}.
\]

Since the mechanism's value is nonnegative and \(E\) has probability
bounded below by a universal positive constant,
\[
    \E_{\cost\sim\mathcal D}
    \left[
        \E_{\mathcal M}
        \bigl[v(\alloc(\cost))\bigr]
        \,\middle|\,
        E
    \right]
    =
    O\left(\frac1{H_m}\right).
\]
Hence, there exists a deterministic cost profile
\(\cost^\star\in E\) such that
\[
    \E_{\mathcal M}
    \bigl[v(\alloc(\cost^\star))\bigr]
    =
    O\left(\frac1{H_m}\right),
\]
while the \(\mms\) benchmark of the corresponding instance is at
least \(1\). The approximation ratio on this instance is therefore
$    \Omega(H_m)
    =
    \Omega(\log m)
    =
    \Omega(\log n).
$
For general \(n\geq 6\), let
$
    m=\left\lfloor\frac{n-2}{2}\right\rfloor.
$
The construction uses \(2m+2\) agents; if \(2m+2<n\), add a dummy
agent with zero cost, zero revenue, and zero contribution to the
valuation. IR in expectation makes the dummy agent's expected payment
nonnegative, so its presence can only tighten the budget constraint
and does not affect the argument. Since \(m=\Theta(n)\), the same
\(\Omega(\log n)\) lower bound follows. 
Values of \(n\) are immaterial to the asymptotic statement.
\end{proof}

\subsection{The \ladderbafo\ Mechanism} 
Our main result in this section (\cref{thm:constant-factor-approx-MMS}) is a 
mechanism that achieves
a constant-factor approximation to $\mms$ under \emph{every} subgame perfect equilibrium (SPE).
\begin{algorithm}[ht]
\caption{\ladderbafo}\label{alg:alt_wish_bafo}
\SetAlgoLined
\KwIn{Partition $(S_1,S_2)$; revenues $r_i$; upper bound of the optimal value $\sumval$; accuracy parameter $\varepsilon\in(0,1)$; failure probability $\delta\in(0,1)$; starting wish $\kappa>0$ }
\KwOut{Allocation $x\in[0,1]^n$, payments $p\in\mathbb{R}^n$.}

\BlankLine
\textbf{Parameters:}
Set $\eta \gets \varepsilon/2$\;
Set $L \gets \left\lceil \log_{1+\eta}(\sumval/\kappa)\right\rceil+2$\;
Set $q \gets \delta/L$\;
Set $T_1\gets \kappa$ and $T_{t+1}\gets (1+\eta)\,T_t$ for $t\ge 1$\;

\BlankLine
\textbf{State:}
Initialize incumbent outcome $(\bar \alloc,\bar b)\gets (0,0)$\;

\BlankLine
\For{$t=1,2,\dots,L$}{
    Let $G_t \gets S_1$ if $t$ is odd, else $G_t \gets S_2$\;
    \tcp{Round $t$: run Sequential BAFO on $G_t$ with the wish predicate $\Phi_{G_t,T_t}$.}
    Announce the selection rule $f_{G_t,T_t}$\;
    Approach sellers in $G_t$ sequentially (in a public fixed order); collect bids $b^t$.
    
    Let $\alloc^t \gets f_{G_t,T_t}(b^t)$.

    \uIf{$\alloc^t = 0$}{
        \tcp{Wish failed: implement the previous successful incumbent.}
        Output $\alloc\gets \bar \alloc$ and payments $p_i \gets \bar b_i\,\bar x_i$ for all $i$; \textbf{stop}.
    }\Else{
        \tcp{Wish succeeded: update the incumbent.}
        Set $(\bar \alloc,\bar b)\gets (\alloc^t,b^t)$.
        
        With probability $q$, output $\alloc\gets \bar \alloc$ and $p_i\gets \bar b_i\,\bar x_i$ for all $i$; \textbf{stop}.
        
        Otherwise (with probability $1-q$), continue to round $t+1$.
    }
}
\BlankLine
\tcp{(This line is never reached if $T_L>\sumval$, but is included for completeness.)}
Output $\alloc^* \gets \bar \alloc$ and payments $p_i \gets \bar b_i\,\bar x_i$ for all $i$.
\end{algorithm}

\begin{theorem}\label{thm:constant-factor-approx-MMS}
The \ladderbafo\ achieves a constant approximation to $\mms$
\emph{in expectation} (over the mechanism's internal randomization) in every subgame perfect equilibrium (SPE) in the worst-case.
\end{theorem}

To understand the need for the ladder-like structure of our
mechanism, it is useful to understand how simpler alternatives
might fail. For instance, suppose we know the optimal partition. Then, a naive approach might simply run a simultaneous competition between these two sets, asking them to generate the maximum value. However, if one set is inherently stronger (i.e., capable of generating significantly more value), the weaker set anticipates defeat and exits the market or bids passively. Without competition, the revenue generators within the stronger set could again become a monopoly and extract all the surplus, driving value to zero.

To overcome this, \ladderbafo\ employs an \emph{alternating ladder} structure. The mechanism operates in rounds, alternating between the two sets. In each round, the active set is asked to satisfy a value target $T$. The first-round target is small and increases multiplicatively in each iteration. By starting with small targets, the mechanism keeps the weaker set engaged and competing as long as the target is within their feasible range, preventing the stronger set from immediately monopolizing the market. Crucially, to incentivize the weaker set to remain active even though it knows that it will not be able to compete beyond a certain threshold, the mechanism utilizes \emph{probabilistic termination}. After every successful round, the mechanism terminates with probability $q$, finalizing the current allocation. This leaves  ``crumbs'' of opportunity: agents in the active set know that fulfilling the current target might lead to an immediate win. This effectively decouples the rounds and forces agents to maximize their utility locally.

\paragraph{Analysis Roadmap.}The proof of \cref{thm:constant-factor-approx-MMS} proceeds in two steps. First, we analyze the mechanism's performance on a \emph{fixed} partition $(S_1, S_2)$. We establish that in any SPE, the mechanism guarantees a value comparable to the maximum feasible value of the weaker set (i.e., the set with lower potential value) (\cref{thm:fixed_partition_extract}). Second, we show that a \emph{random} partition of the market is sufficiently balanced. Specifically, we prove in \cref{lem:random-split} that, in expectation, the weaker half of a random partition still contains enough potential value to approximate the MMS benchmark. Combining these two results yields the main theorem.

\paragraph{Extracting the minimum value of a fixed partition.}

First we show that for any fixed set partition of agents $(S_1,S_2)$ and
 for every $\varepsilon>0$, the principal can use \ladderbafo\
to extract value arbitrarily close to $\min\{\wel(S_1),\wel(S_2)\}$ in \emph{every} SPE.

\begin{lemma}\label{thm:fixed_partition_extract}
Fix an instance $\mathcal{I}=(v,\revenue,\cost)$ and a partition $(S_1,S_2)$ of the agents. Let
$
\wel_{\min} \;:=\; \min\left\{\wel(S_1),\wel(S_2)\right\}.
$
For every $\varepsilon,\delta\in(0,1)$ and every $\kappa\in(0,\wel_{\min})$, in every SPE the final allocation $\alloc^*$ constructed by the \ladderbafo~ satisfies 
\[
\Pr\!\left[
v(\alloc^*)
\;\ge\;
(1-\varepsilon)\,\wel_{\min}
\right]
\;\ge\;
1-\delta,
\]
where the probability is over the mechanism's internal randomization.
\end{lemma}

We first formally define the threshold predicate we use in \cref{alg:alt_wish_bafo} and observe that it satisfies the price monotonicity property.

\begin{definition}[Threshold Predicate]\label{def:threshold-predicate}
    Fix a set $G\subseteq N$ and target $T\ge 0$. Define the predicate $\Phi_{G,T}(\alloc,\bids)$ as follows:
\[
\Phi_{G,T}(\alloc,\bids)=\text{True}
\quad\Longleftrightarrow\quad 
\Bigl(x_i = 0\ \forall i\notin G\Bigr)\ \land\
\Bigl(\sum_{i\in G} b_i x_i \le \sum_{i\in G} r_i x_i\Bigr)\ \land\
\Bigl(v(\alloc)\ge T\Bigr).
\]
\end{definition}

\begin{lemma}\label{lem:phi_threshold_monotone}
For every $G\subseteq N$ and $T\ge 0$, the predicate $\Phi_{G,T}$ is price monotonic~(\cref{def:monotone}).
\end{lemma}
\begin{proof}
Fix any allocation  $\alloc$  and bids $\mathbf{b},\mathbf{b}'$ such that $\Phi_{G,T}(\alloc,\mathbf{b})=\textnormal{True}$
and $b'_i\le b_i$ for all $i$ with $x_i>0$.
The constraint $v(x)\ge T$ depends only on $\alloc$.
The budget constraint $\sum_{i\in G} b_i x_i \le \sum_{i\in G} r_i x_i$ can only become easier when bids of winners decrease.
Therefore $\Phi_{G,T}(\alloc,\mathbf{b}')=\textnormal{True}$.
\end{proof}

Next, we show a result that holds in every round of  \cref{alg:alt_wish_bafo}: if $T_t \le \wel(G_t),$ then in every SPE (conditional on reaching round $t$), the output of round $t$ satisfies $\alloc^t\neq 0$ and $v(\alloc^t)\ge T_t.$ 
The main idea is that the game terminates in the current round with positive probability, and therefore agents that can support the wish $T_t$ have an incentive to do so. The formal proof is 
deferred to \cref{apx:mmslowerbound}.

\begin{lemma}\label{lem:round_success_fixed_partition}
Fix a round $t$ of \ladderbafo\ with active set $G_t$ and threshold $T_t$, and condition on any history that reaches the start of round $t$.
If $T_t \le \wel(G_t)$, then in every SPE, the BAFO call of round $t$ outputs an allocation $\alloc^t\neq \mathbf{0}$
with $v(\alloc^t)\ge T_t$.
\end{lemma}

\begin{proof}[Proof Sketch]
    Recall that the round-$t$ outcome is \emph{implemented with positive probability $q$}, so agents have a strict incentive not to destroy feasibility. If the wish fails (i.e., $\alloc^t=\mathbf{0}$), \cref{alg:alt_wish_bafo} stops immediately and implements the previous incumbent, so no agent in $G_t$
can benefit (in round $t$) from forcing failure.
If the wish succeeds, the mechanism terminates and implements this outcome $(\alloc^t,b^t)$ with probability $q>0$.
Therefore, whenever the target is feasible under true costs (equivalently $T_t\le \wel(G_t)$), any agent who can profit from a feasible round-$t$ outcome can get strictly more utility
by keeping the wish feasible, since doing so yields strictly positive expected utility from the $q$ termination event; sellers who cannot profit cannot sabotage feasibility either, because,
by price monotonicity, feasibility can be witnessed by an allocation that does not rely on them (or relies on them only when they bid $c_i$).
\end{proof}

We are now ready  to prove \cref{thm:fixed_partition_extract}.

\begin{proof}[Proof of \cref{thm:fixed_partition_extract}]
Fix $\varepsilon,\delta\in(0,1)$ and run \cref{alg:alt_wish_bafo} with
$\eta=\varepsilon/2$, starting wish $\kappa\in(0,\wel_{\min})$,
horizon $L=\lceil \log_{1+\eta}(\sumval/\kappa)\rceil+2$, and stopping probability $q=\delta/L$.
If $\wel_{\min}=0$, the claim is trivial.

Fix any round $t$ and condition on reaching its start.
If $T_t\le \wel(G_t)$, then by Lemma~\ref{lem:round_success_fixed_partition} the round-$t$ BAFO call outputs $x^t\neq 0$ and $v(x^t)\ge T_t$
in every SPE of the continuation game.

Consider the event that the mechanism does not stop due to the coin toss before the first failed wish.
Let $\tau$ be the first round in which the wish fails, i.e., $x^\tau=0$.
By the preceding paragraph, failure implies $T_\tau>\wel(G_\tau)\ge \wel_{\min}$.
Therefore $T_{\tau-1}=T_\tau/(1+\eta)>\wel_{\min}/(1+\eta)$.
When round $\tau$ fails, the mechanism implements the incumbent from round $\tau-1$, whose value is at least $T_{\tau-1}$.
Hence on this event,
\[
v(\alloc^*) \;\ge\; \frac{\wel_{\min}}{1+\eta}.
\]

The stopping coin is tossed only after successful rounds, and there are at most $L$ rounds.
Thus, under any strategy profile (in particular any SPE),
\[
\Pr[\text{stop before the first failed wish}]
\;\le\; \sum_{t=1}^L q \;=\; qL \;=\; \delta.
\]
So with probability at least $1-\delta$ we reach the first failed wish and the bound
$v(\alloc^*) \geq \nicefrac{\wel_{\min}}{(1+\eta)}$
holds.
Since $\eta=\varepsilon/2$, we have $(1+\eta)^{-1}\ge 1-\varepsilon$, and therefore
\[
\Pr\!\left[v(\alloc^*) \ge (1-\varepsilon)\,\wel_{\min}\right] \;\ge\; 1-\delta. \qedhere
\]
\end{proof}

\paragraph{Constant Approximation of the Worst-Group.} 
The last ingredient we need for the proof 
of \cref{thm:constant-factor-approx-MMS}
is to show that
when we draw a partition uniformly at random from the set of 
all partitions, then, in expectation, the value of
the worst
of the two sets approximates
\mms\ within a constant factor. The proof 
relies on a careful accounting argument.

\begin{lemma}\label{lem:random-split}
 For any monotone, fractional subadditive valuation function $v$, the expected value achieved by the worst of the two sets
of a randomly drawn partition satisfies
\[
    \mathbb{E}_{(P_1, P_2)} \left[ \min(\mathcal{W}(P_1), \mathcal{W}(P_2)) \right] \ge \frac{1}{40} MMS \,,
\]
where $(P_1, P_2)$ is drawn uniformly at random from the set of all partitions.
\end{lemma}

\begin{proof}
Let $(S_1^*, S_2^*)$ be a fixed partition achieving the MMS benchmark. Without loss, we analyze the performance on a single witness set $S^* \in \{S_1^*, S_2^*\}$ by conditioning on the structure of its optimal budget-balanced allocation vector, $\mathbf{x}^*$. 

We decompose the value of the optimal allocation $v(\mathbf{x}^*)$ based on the agent types in its support. Let $C$ be the set of funders ($r_i \ge c_i$) and $D$ be the set of creators ($r_i < c_i$) utilized in $\mathbf{x}^*$. By subadditivity,
\[
v(\mathbf{x}^*)
=
v\!\left(\mathbf{x}^*|_C+\mathbf{x}^*|_D\right)
\leq
v(\mathbf{x}^*|_C)+v(\mathbf{x}^*|_D),
\]

where $v(\mathbf{x}^*|S)$ denotes the value of the allocation restricted to agents in set $S$. Based on this decomposition, one of the following two cases must hold:

\textbf{Case 1:  ($v(\mathbf{x}^*|_C) \ge \frac{1}{5} v(\mathbf{x}^*)$).} For every realized partition $(A,B)$, subadditivity gives
$
v(\mathbf{x}^*|_C)
\leq
v(\mathbf{x}^*|_{A\cap C})
+
v(\mathbf{x}^*|_{B\cap C}).
$
By
symmetry, the event
$
v(\mathbf{x}^*|_{A\cap C})
\geq
\frac{1}{2}v(\mathbf{x}^*|_C)
$
occurs with probability at least $1/2$. Since funders can be self-funded, this allocation is budget-feasible.$$ \mathcal{W}(A) \ge \frac{1}{2} v(\mathbf{x}^*|_C) \ge \frac{1}{10} v(\mathbf{x}^*). $$

\textbf{Case 2: ($v(\mathbf{x}^*|_D) \ge \frac{4}{5} v(\mathbf{x}^*)$).} 
In this scenario, we establish a lower bound by constructing a specific candidate allocation for the random set $A$. Let $\sigma_{\text{total}} = \sum_{i \in C} (r_i - c_i)x_i^*$ be the total budget surplus generated by the optimal allocation on the funders. Let $\sigma_A = \sum_{i \in A \cap C} (r_i - c_i)x_i^*$ be the portion of this surplus captured by set $A$. 
We define two independent events for the random set $A$:
\begin{enumerate}
\item Surplus Capture: $A$ captures at least half the total surplus ($\sigma_A \ge \frac{1}{2} \sigma_{\text{total}}$). (Prob $\ge 1/2$).
\item Value Capture: $A$ captures at least half the value of the creator set ($\mathbf{x}^*|_{A \cap D} \ge \frac{1}{2} v(\mathbf{x}^*|_D)$). To see that this event occurs with probability at least $1/2$,
    observe that subadditivity gives
    $
    v(\mathbf{x}^*|_D)
    \leq
    v(\mathbf{x}^*|_{A\cap D})
    +
    v(\mathbf{x}^*|_{B\cap D}).
    $
    Thus, at least one of $A$ or $B$ captures at least half of the
    creator value, and the claim follows by symmetry.\end{enumerate}

    Since the surplus is derived from agents in $C$ and the value is derived from agents in
$D$, and these sets are disjoint, the two events above are independent. Hence,
their intersection occurs with probability at least $1/4$.

On the intersection of these events, let
$
\alpha
=
\frac{\sigma_A}{\sigma_{\text{total}}}.
$
Notice that $\alpha\geq 1/2$. We consider the allocation
\[
\mathbf{y}
=
\mathbf{x}^*|_{A\cap C}
+
\alpha\mathbf{x}^*|_{A \cap D}.
\]

We first verify that $\mathbf{y}$ is budget-feasible for $A$. Let
$
\delta_{\text{total}}
=
\sum_{i\in D}(c_i-r_i)x_i^*
$ and 
$\delta_A
=
\sum_{i\in A\cap D}(c_i-r_i)x_i^*.
$
Since $\mathbf{x}^*$ is budget-feasible,
$
\delta_{\text{total}}
\leq
\sigma_{\text{total}}.
$
Moreover, in Case 2 we have
$v(\mathbf{x}^*|_D)>0$. Since $v(\mathbf{0})=0$, this implies that
$\delta_{\text{total}}>0$, and therefore
$\sigma_{\text{total}}>0$.

The net revenue generated by $\mathbf{y}$ is
\[
\sum_i(r_i-c_i)y_i
=
\sigma_A-\alpha\delta_A
\geq
\sigma_A-\alpha\delta_{\text{total}}
\geq
\sigma_A-\alpha\sigma_{\text{total}}
=0.
\]
Thus, $\mathbf{y}$ is budget-feasible for $A$.

By monotonicity,
$
\wel(A)
\geq
v(\mathbf{y})
\geq
v(\alpha\mathbf{x}^*|_{A \cap D}).
$
Since $v$ is subadditive,
\[
v(\mathbf{x}^*|_{A \cap D})
=
v\left(
\frac{1}{2}\mathbf{x}^*|_{A \cap D}
+
\frac{1}{2}\mathbf{x}^*|_{A \cap D}
\right)
\leq
2v\left(\frac{1}{2}\mathbf{x}^*|_{A \cap D}\right). \qquad 
\Rightarrow \qquad v\left(\frac{1}{2}\mathbf{x}^*|_{A \cap D}\right)
\geq
\frac{1}{2}v(\mathbf{x}^*|_{A \cap D}).
\]

Since $\alpha\geq1/2$, monotonicity further implies
\[
v(\alpha\mathbf{x}^*|_{A \cap D})
\geq
v\left(\frac{1}{2}\mathbf{x}^*|_{A \cap D}\right)
\geq
\frac{1}{2}v(\mathbf{x}^*|_{A \cap D}).
\]
Substituting the bound from the Value Capture event gives
\[
\wel(A)
\geq
\frac{1}{2}v(\mathbf{x}^*|_{A \cap D})
\geq
\frac{1}{2}\cdot\frac{1}{2}
v(\mathbf{x}^*|_D)
=
\frac{1}{4}v(\mathbf{x}^*|_D).
\]
Finally, using the case assumption that creators hold the majority of
the value,
$
v(\mathbf{x}^*|_D)
\geq
\frac{4}{5}v(\mathbf{x}^*),
$
we have
\[
\wel(A)
\geq
\frac{1}{4}\cdot\frac{4}{5}v(\mathbf{x}^*)
=
\frac{1}{5}v(\mathbf{x}^*).
\]

We determine the expected minimum value $\mathbb{E}[\min(\wel(A), \wel(B))]$ by analyzing the three possible configurations of the witness sets $(S_1^*, S_2^*)$.
\begin{itemize}\item \textbf{Both Case 1.} For success, $A$ must approximate one set and $B$ the other. Since Case 1 has success probability $1/2$, the probability of a valid assignment is $1/2$. The value is limited by the Case 1 guarantee ($1/10$).$$\mathbb{E}[\min(\wel(A), \wel(B))] \ge \frac{1}{2} \cdot \frac{1}{2} \cdot 2 \times \frac{1}{10} \text{MMS} = \frac{1}{20} \text{MMS}.$$\item \textbf{One Case 1, One Case 2.}
Assume $S_1^*$ is Case 1 (Funder) and $S_2^*$ is Case 2 (creator).
A "Global Win" requires $A$ to succeed on $S_1^*$ ($P=1/2$) AND $B$ to succeed on $S_2^*$ ($P=1/4$), or vice-versa.
\begin{itemize}
    \item $\Pr[A \text{ wins } S_1^*, B \text{ wins }S_2^*] = 1/2 \times 1/4 = 1/8.$
    \item $\Pr[B \text{ wins } S_1^*, A \text{ wins }S_2^*] = 1/2 \times 1/4 = 1/8$.
    \item Total Probability = $1/4$.
\end{itemize}
The value is the minimum of the two guarantees: $\min(\frac{1}{10}, \frac{1}{5}) = \frac{1}{10}$.
$$\mathbb{E}[\min] \ge \frac{1}{4} \times \frac{1}{10} \text{MMS} = \frac{1}{40} \text{MMS}.$$

\item \textbf{Both Case 2.}
Both sets require the stricter conditions of Case 2 ($P=1/4$).
\begin{itemize}
    \item $\Pr[ A \text{ wins } S_1^*, B \text{ wins } S_2^*] = 1/4 \times 1/4 = 1/16$.
    \item $\Pr[B \text{ wins } S_1^*, A \text{ wins } S_2^*] = 1/4 \times 1/4 = 1/16$.
    \item Total Probability = $1/8$.
\end{itemize}
The value is limited by the Case 2 guarantee ($1/5$).
$$\mathbb{E}[\min] \ge \frac{1}{8} \times \frac{1}{5} \text{MMS} = \frac{1}{40} \text{MMS}.$$
\end{itemize}In all scenarios, the expected minimum value is at least $\frac{1}{40} \text{MMS}$.\end{proof}

We are now ready to prove the guarantee of the \ladderbafo\ mechanism.

\begin{proof}[Proof of \cref{thm:constant-factor-approx-MMS}]
Draw the partition $(S_1,S_2)$ uniformly at random and run \ladderbafo\ on $(S_1,S_2)$.
For any fixed partition, \cref{thm:fixed_partition_extract} implies that in every SPE the output satisfies
$v(\alloc^*) \ge (1-\varepsilon)\,\min\{\wel(S_1),\wel(S_2)\}$ with probability at least $1-\delta$.
Taking expectation over the mechanism's randomness and then over the random partition, we get
\[
\mathbb{E}\!\left[v(\alloc^*)\right]
\;\ge\;
(1-\varepsilon)(1-\delta)\cdot
\mathbb{E}_{(S_1,S_2)}\!\left[\min\{\wel(S_1),\wel(S_2)\}\right].
\]
By Lemma~\ref{lem:random-split}, the final expectation is at least $\mms/40$, yielding
$\mathbb{E}[v(\alloc^*)]\ge (1-\varepsilon)(1-\delta)\,\mms/40$.
Choosing constant $\varepsilon,\delta$ gives the claimed constant-factor approximation.
\end{proof}

%% file: conclusion.tex
\section{Conclusion}\label{sec:conclusion}
In this work, we introduce the problem of designing mechanisms that incentivize strategic agents to form self-funded marketplaces. This problem generalizes the well-studied budget-feasible mechanism design problem, where the requirement is that total payments are bounded by a publicly known budget $B$; in our setting, the ``budget'' is endogenous and can vary depending on the selected set of agents $S$.

To evaluate the performance of our mechanisms, we first consider the first-best benchmark, $\mathrm{OPT}$. We show that truthful mechanisms cannot guarantee any bounded approximation, even when the value of $\mathrm{OPT}$ is known in advance. Furthermore, we show that even in a restricted setting where all agent costs are public except one, and the designer knows that $\mathrm{OPT}\in[1,\mathcal V]$, no mechanism can achieve an approximation ratio better than $\Omega(\log \mathcal V)$ approximation in any equilibrium. We complement these results by proposing a sequential auction that guarantees a $O(\log \mathcal V)$ approximation in \emph{every} subgame-perfect equilibrium if the designer knows that $\mathrm{OPT}$ lies within a range $[1,\mathcal V]$.
We then introduce the Maximin Share (MMS) benchmark, adapted from a well-studied notion in fair division. 
We prove that the approximation ratio of any truthful mechanism relative to MMS is $\Omega(\log n)$, and we design a sequential auction that achieves a constant-factor approximation to MMS in every subgame-perfect equilibrium. 

It is worth noting that while our general model allows for fractional efforts, the majority of our guarantees---including the $O(\log \mathcal{V})$ approximation to $\mathrm{OPT}$ (\cref{thm:guessbafo}) and the ability of BAFO to extract the optimal value of the weaker group in a fixed partition (\cref{thm:fixed_partition_extract})---hold even if allocations are restricted to be integral. The constant-factor approximation to the MMS benchmark, however, currently relies on fractional allocations to bound the expected value of a random partition.

Beyond these approximation guarantees, our results highlight two concepts of broader interest for mechanism design. First, the MMS benchmark provides a robust measure of market competitiveness. By filtering out monopolistic instances and remaining stable in competitive environments, MMS could serve as a valuable benchmark for other prior-free settings where efficiency hinges on competition. Second, our positive results utilize the sequential Best and Final Offer (BAFO) protocol, an auction format widely used in practice but understudied in theory. We show that sequential BAFO acts as a decentralized constraint solver: if a price-monotonic condition is feasible at true costs, sequential BAFO guarantees it in every subgame-perfect equilibrium, avoiding the coordination failures inherent to simultaneous bidding.

This framework also leaves several open questions. One direction is to study
learning dynamics and whether natural adjustment processes converge to the
equilibria analyzed in this paper. Another is to understand what additional
information, such as prior distributions, or
partial information about agents' costs, allows the principal to obtain
stronger guarantees.
\newpage

%% file: apx_bafo.tex
\section{Omitted Details from \cref{sec:bafo}}\label{apx:bafo}

\paragraph{Proof of \cref{thm:learning_bafo}}
\begin{proof}
The proof proceeds in two steps: establishing the feasibility of the target threshold and then applying the wish-come-true property. 

The mechanism sets the target welfare to $T = \hat{\opt} / \bar{\eta}$.
By the definition of the prediction error, we know that $\hat{\opt} \le \bar{\eta} \cdot \opt$.
Dividing by $\bar{\eta}$, we obtain:
\[ T = \frac{\hat{\opt}}{\bar{\eta}} \le \frac{\bar{\eta} \cdot \opt}{\bar{\eta}} = \opt. \]
Thus, the target threshold $T$ is feasible; there exists at least one allocation (specifically, the optimal budget-balanced allocation $\mathbf{x}^*$) that satisfies both the budget constraint and the condition $v(\mathbf{x}^*) \ge T$ under the agents' true costs $\mathbf{c}$.

Let $\Phi(\mathbf{x}, \mathbf{b})$ be the predicate required by the mechanism:
\[ \Phi(\mathbf{x}, \mathbf{b}) \iff \left(\sum x_i b_i \le \sum x_i r_i\right) \land \left(v(\mathbf{x}) \ge T\right). \]
We verify the two conditions required for Theorem~\ref{thm:implement}:
\begin{itemize}
    \item \textbf{price monotonicity:} The budget constraint $\sum x_i b_i \le \sum x_i r_i$ is clearly price monotonic (if it holds for $\mathbf{b}$, it holds for any $\mathbf{b}' \le \mathbf{b}$). The value constraint $v(\mathbf{x}) \ge T$ depends only on the allocation, so it is unaffected by price changes. Thus, their conjunction $\Phi$ is Price Monotonic.
    \item \textbf{$\Phi$-consistency:} The mechanism's selection rule explicitly maximizes welfare over the set of allocations satisfying $\Phi$. If this set is non-empty, it selects a valid member; otherwise, it outputs $\mathbf{0}$. This satisfies the definition of a $\Phi$-Consistent rule.
\end{itemize}

Since $T \le \opt$, the premise of Theorem~\ref{thm:implement} is met (a valid allocation exists at true costs). Therefore, in every SPE, the agents coordinate to produce an allocation $\mathbf{x}$ such that $\Phi(\mathbf{x}, \mathbf{b})$ is true.

The validity of the outcome implies $v(\mathbf{x}) \ge T$. Substituting the definition of $T$:
\[ v(\mathbf{x}) \ge \frac{\hat{\opt}}{\bar{\eta}}. \]
Using the other side of the prediction error bound ($\hat{\opt} \ge \opt / \bar{\eta}$), we conclude:
\[ v(\mathbf{x}) \ge \frac{1}{\bar{\eta}} \left( \frac{\opt}{\bar{\eta}} \right) = \frac{\opt}{\bar{\eta}^2}. \qedhere \]
\end{proof}

%% file: apx_mms.tex
\section{Omitted Details from \cref{sec:mms}} \label{apx:mms}

\subsection{Comparison of prior-free benchmarks}\label{apx:benchmark}
In this section, we motivate our choice of the Maximin-Share (MMS) benchmark by contrasting it with two other standard benchmarks from the mechanism design literature: the \emph{Frugality} benchmark and the $\mathcal{F}^{(2)}$ benchmark from competitive auctions. We demonstrate why these alternatives are ill-suited for the budget-balanced procurement setting, where agent contributions are two-fold (revenue and value).

\paragraph{The Frugality Benchmark.}In the context of frugal mechanism design (e.g., path auctions), the benchmark is typically defined by the best solution of the system when the optimal solution is removed. Adapted to our welfare maximization setting, the Frugality benchmark corresponds to the maximum budget-balanced welfare achievable using only agents \emph{not} included in the optimal solution $S^*$. Formally:
$$ \text{Frugal}(\mathcal{I}) = \max_{\alloc: x_i = 0 \text{ for all } i \in S^*} \left\{ v(\alloc) \mid \sum_{i \in S} b_i \le \sum_{i \in S} r_i \right\}. $$
In the budget-balanced setting, this benchmark is frequently trivial (i.e., zero). The optimal solution $S^*$ fully allocates all the ``funders'' (agents with $r_i \geq c_i$). By definition, removing $S^*$ eliminates the source of budget surplus. The remaining agents, $N \setminus S^*$, lack the revenue to form a non-negative feasible budget-balanced set, making the benchmark zero.

\paragraph{The $\mathcal{F}^{(2)}$ Benchmark.}
The $\mathcal{F}^{(2)}$ benchmark, standard in digital goods and competitive auctions, is defined as the optimal welfare achievable using a single uniform price $p$, subject to the constraint that at least two agents are selected. Formally:
$$ \mathcal{F}^{(2)} = \max_{p \ge 0} \left\{ \sum_{i: b_i \le p} v_i \;\middle|\; |\{i : b_i \le p\}| \ge 2 \land \sum_{i: b_i \le p} p \le \sum_{i: b_i \le p} r_i \right\}. $$
While robust for homogeneous goods where agents differ only in valuation (or cost), $\mathcal{F}^{(2)}$ is a poor fit for our problem due to the heterogeneity of agent contributions. Agents in our setting contribute in two distinct dimensions: \emph{revenue generation} ($r_i$) and \emph{welfare generation} ($v_i$).\begin{enumerate}\item \textbf{Incompatibility with Uniform Pricing:} A single uniform price imposes a symmetric constraint on asymmetric roles. To maximize efficiency, funders should ideally be paid close to their cost (to maximize the subsidy available for others), while creators should be paid as much as the surplus allows (to ensure their participation). A uniform price cannot simultaneously minimize the cost of revenue generation and maximize the payout to value creators.\item \textbf{Type Uncertainty:} One might argue that we can impose ``uniformity'' on two different dimensions, both value and revenue, depending on the contribution type of a given agent. However, without prior information, the mechanism cannot distinguish which ``type'' (revenue-heavy vs. value-heavy) an agent belongs to. 
\end{enumerate}
Consequently, $\mathcal{F}^{(2)}$ often serves as an arbitrarily loose lower bound that does not accurately reflect the achievable welfare in a budget-balanced environment.

\subsection{Missing Proofs from \cref{sec:mms}}\label{apx:mmslowerbound}

\begin{lemma}\label{lem:round_incentive}
Fix a round $t$ and condition on reaching it.
Consider any seller $i\in G_t$ at some history within the round-$t$ BAFO call.
If there exists a bid $b_i>c_i$ such that, under the equilibrium continuation within the round,
the resulting terminal outcome satisfies $\Phi_{G_t,T_t}(\alloc^t,\bids^t)=\textnormal{True}$ and allocates $\alloc^t_i>0$,
then at that history seller $i$ has a strict incentive to choose some action that preserves feasibility (i.e., does not lead to $\alloc^t=\mathbf{0}$).
\end{lemma}
\begin{proof}
If seller $i$ chooses such a bid and feasibility is preserved, then with probability $q>0$ the mechanism stops after round $t$
and seller $i$ receives expected utility at least $q\cdot \alloc^t_i(b_i-c_i)>0$.
If instead seller $i\!$ plays an action that causes $\alloc^t=\mathbf{0}$, the algorithm stops immediately and implements the previous feasible outcome,
in which seller $i$ is, by design, not selected, and hence receives utility $0$.
Therefore feasibility-preserving behavior dominates other actions.
\end{proof}

\begin{proof}[Proof of \cref{lem:round_success_fixed_partition}]
Since $T_t\le \wel(G_t)$, there exists $ \alloc^*$ supported on $G_t$ such that
$\sum_{i\in G_t} c_i  \alloc^*_i \le \sum_{i\in G_t} r_i  \alloc^*_i$ and $v( \alloc^*)\ge T_t$,
i.e., $\Phi_{G_t,T_t}( \alloc^*,\mathbf{c})=\textnormal{True}$.
Moreover, by Lemma~\ref{lem:phi_threshold_monotone}, $\Phi_{G_t,T_t}$ is price monotonic.

Consider the sequential BAFO game induced in round $t$ with selection rule $f_{G_t,T_t}$.
We show by backward induction over the bidding order in $G_t$ that along any SPE path, feasibility of $\Phi_{G_t,T_t}$
cannot be destroyed once it is possible at the conditional prices.

At any history where seller $i$ moves, there are two cases.
If seller $i$ can obtain strictly positive utility in a feasible terminal outcome, then by Lemma~\ref{lem:round_incentive}
seller $i$ strictly prefers actions that keep feasibility over any action that makes $\alloc^t=\mathbf{0}$.
If seller $i$ cannot obtain strictly positive utility from feasibility, then either (a) feasibility does not require $i$ (so there exists a feasible allocation with $\alloc_i=0$),
in which case by price monotonicity seller $i$'s bid cannot destroy feasibility, or (b) $i$ can be included only at $b_i=c_i$,
in which case bidding $c_i$ preserves feasibility and never yields negative utility.

Thus, at every step, equilibrium play preserves feasibility whenever it is achievable.
Since feasibility is achievable at true costs, the terminal outcome of round $t$ satisfies
$\Phi_{G_t,T_t}(\alloc^t,\bids^t)=\textnormal{True}$ by $\Phi$-consistency of $f_{G_t,T_t}$.
In particular, $v(\alloc^t)\ge T_t$, and hence $\alloc^t\neq \mathbf{0}$.
\end{proof}